\newif\ifSubmission

\newif\ifComments
\newif\ifAnonymous

\ifSubmission
    \Commentsfalse
    \Anonymoustrue
\else
    \Anonymousfalse
    \Commentsfalse
\fi

\ifSubmission
    \documentclass{llncs}
\else
    \documentclass[11pt]{article}
    \usepackage{libertine}
    \usepackage{fullpage}
    \usepackage[letterpaper,top=2cm,bottom=2cm,left=3cm,right=3cm,marginparwidth=1.75cm]{geometry}
\fi
\usepackage{graphicx}
\usepackage{braket}
\usepackage{amsmath,bbm}
\usepackage{amssymb}
\ifSubmission\else
    \usepackage{amsthm}
\fi
\ifSubmission
    \newcommand{\submissionbreak}[1]{\\& #1}
\else
    \newcommand{\submissionbreak}[1]{}
\fi
\usepackage[dvipsnames]{xcolor}
\usepackage[colorlinks=true, allcolors=blue]{hyperref}
\usepackage{cleveref}
\usepackage{verbatim}
\usepackage{physics}

\usepackage{booktabs}
\usepackage{tabularx}
\usepackage{array}
\usepackage{caption}
\usepackage{hyphenat}

\ifSubmission\else
    \newtheorem{theorem}{Theorem}
    \newtheorem{lemma}{Lemma}
    
    \newtheorem{definition}{Definition}
    
    \newtheorem{corollary}{Corollary}

\fi
\newtheorem{result}{Result}
\newtheorem{construction}{Construction}

\ifComments
    \newcommand{\authnote}[3]{\textcolor{#3}{[{\footnotesize {\bf #1:} { {#2}}}]}}
\else
    \newcommand{\authnote}[3]{}
\fi

\newcommand{\bit}{\{0,1\}}
\newcommand{\negl}{\mathsf{negl}}
\newcommand{\getsr}{\overset{\$}{\leftarrow}}
\newcommand{\secp}{\lambda}

\newcommand{\poly}{\mathsf{poly}}

\newcommand{\bbE}{\mathbb{E}}

\newcommand{\bbN}{\mathbb{N}}

\newcommand{\calA}{\mathcal{A}}

\newcommand{\bits}{\{0,1\}}
\newcommand{\cC}{\mathcal C}
\newcommand{\cW}{\mathcal W}
\newcommand{\cY}{\mathcal Y}

\newcommand{\Acc}{\mathsf{Acc}}

\newcommand{\Succ}{\mathsf{Succ}}
\newcommand{\wt}{\operatorname{wt}}
\newcommand{\RO}{H}

\newcommand{\Exp}{\mathbb E}
\newcommand{\Prob}{\Pr}

\newcommand{\Havg}{H^{\mathrm{avg}}_{\infty}}

\newcommand{\Prove}{\mathsf{Prove}}

\newcommand{\Verify}{\mathsf{Verify}}

\ifAnonymous
   \author{}
\ifSubmission
    \institute{}
\fi
\else
    \author{
    Siddhartha Jain\thanks{UT Austin. Email: \texttt{sidjain@utexas.edu}} 
    \quad Saachi Mutreja\thanks{Columbia University. Email: \texttt{sm5540@columbia.edu}}
    \quad Bhaskar Roberts\thanks{UC Berkeley and NTT Research. Email: \texttt{bhaskarroberts@gmail.com}}}
\fi

\title{Certified Randomness with Optimal Rate}

\date{}

\begin{document}

\maketitle

\begin{abstract}
    The generation of certified random bits is an emerging near-term application of quantum computers. Potential applications, such as randomness beacons and CRS generation, require nearly uniform randomness, whose rate (the ratio of min-entropy to bitlength) is $\sim 1$. However, existing protocols for certified randomness either produce weakly random strings with rate $o(1)$, or they require a trusted random seed.
    
    We show how to certify randomness with the optimal rate $\sim 1$ without requiring any trusted randomness from the verifier. Our protocol is secure unconditionally in the quantum random oracle model (QROM). This improves the result of Coladangelo et al.~\cite{CKMRST26}, who certified \emph{weakly} random strings unconditionally in the QROM, and answers a question posed by Aaronson and Hung~\cite{AH23}.
    
    We also define and construct a proof of \emph{conditional} min-entropy, which certifies optimal min-entropy even conditioned on an adversarially chosen transcript. We show an application of this primitive to randomness beacons, where each pulse should have high min-entropy conditioned on all previous messages.
\end{abstract}

\ifSubmission\else
\fi

\section{Introduction}
Trusted randomness is a scarce and valuable resource. It is used for electronic voting, generation of common random strings, and leader election for proof-of-stake blockchains~\cite{Adi08,ACC+25,GHMVZ17,CMB23}. These applications are only secure if the randomness was sampled honestly. In a trustless system, any such randomness must be verifiable; i.e. parties need the ability to verify that the generator's output is random.

This is the purpose of quantum certified randomness, which allows a verifier with limited randomness to check that a longer random string sampled by an untrusted prover has high min-entropy. 
The hope is to use certified randomness for sortition, randomness beacons, CRS generation, etc.~\cite{ACC+25}. However, these applications require the certified string to be uniformly random or nearly so.

Existing protocols for certified randomness are insufficient because they only guarantee weakly random sources, which may produce bad outputs with high probability. For example, \cite{CKMRST26} Construction 6.2 produces an $n$-bit string with $h = o\left(n^{1/4}\right)$ bits of certified min-entropy, so the string's support may be a negligible fraction of $\bit^n$. \Cref{tab:min-entropy-comparison} shows the rate (ratio of certified min-entropy $h$ to bitlength $n$) of several existing proofs of min-entropy. Our goal is a proof of min-entropy that guarantees uniform, or nearly uniform, randomness; i.e. the rate should be $\sim 1$.

Prior work improves the rate with seeded extractors, but this requires the verifier to sample a short seed honestly. \cite[Theorem~3.7]{YZ24} shows that if the verifier samples a short seed of length $\poly[\secp,\log(n)]$ bits, they can certify uniform randomness in a longer string of length $n$ by running the prover's output through a seeded extractor. This approach is unsatisfying because it still assumes that some randomness was generated honestly. The goal of certified randomness is to remove the need for trusted randomness, and to best capture the spirit of this goal, we require our verifier to be deterministic.

We consider the following setting:
    \begin{enumerate}
        \item (\textbf{Deterministic Verifier}) All inputs to the verifier, including any random strings, are provided by the prover. \footnote{Since we work in the random oracle model, we also allow both parties to use a public random oracle. Prior work \cite{YZ24,KRT26,CKMRST26} takes this approach, and it is reasonable because the random oracle may be heuristically replaced with a concrete hash function, such as SHA-2, that can be hardwired into the verifier.} \label{item:deterministic-verifier}
        \item (\textbf{Malicious Prover}) The prover may be fully malicious, deviating arbitrarily from the prescribed protocol. The only constraint is that they run in $\poly(\secp)$ time. 
        \label{item:malicious-prover}
    \end{enumerate}
    We will consider the following question: 
    \begin{quote}
        In the setting described above, what is the highest min-entropy threshold that the verifier can certify?
    \end{quote}
    We will express the threshold in terms of $n$, the length of the certified string, and $\secp$, the security parameter.

    \noindent We answer this question completely by providing matching upper and lower bounds.

First, it is impossible to certify uniform randomness when the verifier is deterministic. This is because the prover could use rejection sampling to fix $O(\log \secp)$ bits of the verifier's output in $\poly(\secp)$ time \cite{AH23,FSW25}. The prover would run the honest prover and verifier protocols to sample from the source and compute the verifier's eventual output. This is possible because the verifier is deterministic. If the prover does not like the output, they can reject the output and try again. In $\poly(\secp)$ trials, they can fix $O(\log \secp)$ bits of the output. Therefore, our upper bound says that for any $\gamma(\secp) = O(\log \secp)$, it is impossible to certify $n - \gamma(\secp)$ bits of min-entropy.

The best result possible is to certify $n - \gamma(\secp)$ bits of min-entropy for all $\gamma(\secp) = \omega(\log \secp)$. We call this the \emph{maximal min-entropy} guarantee, where \emph{maximal} is used to mean \emph{highest possible}. Our \Cref{constr:maxent} does achieve the maximal min-entropy guarantee. We thus achieve one of Aaronson and Hung~\cite{AH23}'s stated goals for future work.

For many applications, the maximal min-entropy guarantee is effectively as good as the guarantee of uniform randomness. \cite{DY13} showed that for cryptographic primitives with ``square security'' (such as CPA-secure encryption, weak PRFs, and unpredictability primitives -- where the adversary's probability of finding an accepted answer is negligible), the challenger's uniform random string can be replaced with a high-min-entropy string (with min-entropy $\geq n - d$), and the adversary's advantage in the security game will only increase from $\varepsilon$ to $\sqrt{2^d \cdot \varepsilon}$. If the advantage is negligible when the challenger receives uniform randomness, then as long as the challenger's string satisfies the maximal min-entropy guarantee above, the advantage will still be negligible.

\subsection{Our Results}
Here we give an informal overview of our results.

\begin{result}[Informal]\label[result]{result:proof-of-maximal-min-entropy}
    We define and construct a proof of maximal min-entropy in the quantum random oracle model. 
\end{result}
A proof of maximal min-entropy allows a quantum prover to convince a deterministic verifier that a given string $x$ has maximal min-entropy. It has the following syntax:
\begin{itemize}
    \item$\Prove^{H}(1^\lambda,1^n) \to \pi:$ A quantum algorithm that outputs a classical proof $\pi$.
    \item$\Verify^{H}(1^{\lambda}, 1^n, \pi) \to x:$ A \emph{deterministic} classical algorithm that verifies $\pi$. If verification passes, it outputs an $n$-bit string $x$ with maximal min-entropy. If verification fails, it outputs $x = \bot$.
\end{itemize}
Both $\Prove$ and $\Verify$ run in polynomial time and get query access to the random oracle $H$.

Our security guarantee says that the min-entropy of $x$ is maximal, even conditioned on the choice of random oracle $H$. More formally, the adversarial prover induces a distribution over their output $\Pi$ and the verifier's output $X \gets \Verify^H(1^\secp, 1^n, \Pi)$. We require that for any malicious polynomial-time prover, either their proof will be rejected ($X = \bot$) with overwhelming probability or $X$ has maximal min-entropy conditioned on $H$:
\[\Pr_H\left[\Pr[X \neq \bot] \geq \frac{1}{\poly(\secp)} \land H_\infty[X | X \neq \bot] \leq n - \omega(\log \secp)\right] = \negl(\secp)\]
Note that $\Pr[X \neq \bot]$ and $H_\infty[X | X \neq \bot]$ implicitly condition on the choice of $H$ because they are inside of $\Pr_H[\cdot]$.

We construct a proof of maximal min-entropy as follows. We use \cite{CKMRST26}'s proof of min-entropy to generate a weakly random source, and then hash the source with a compressing random oracle. The random oracle acts as a condenser, transforming the weakly random source to a strongly random string with min-entropy $\geq n - \omega(\log \secp)$. Note that \cite{CKMRST26}'s construction already requires a random oracle. Therefore we improve \cite{CKMRST26}'s result with no additional assumptions to achieve the highest possible min-entropy guarantee.

These techniques also show that random oracles are good seedless condensers against quantum adversaries. This improves on \cite{FSW25}, who proved a similar claim for classical adversaries.

Our next result enables proofs of \emph{conditional} min-entropy, in which the adversary chooses a label $C$ and certifies that $X$ has maximal min-entropy conditioned on $C$. This is useful in the application to randomness beacons, where the beacon outputs a sequence of unpredictable random pulses, and each pulse should have high min-entropy given the transcript of messages and the pulses up to that point. \Cref{result:proof-of-maximal-conditional-min-entropy} is a strengthening of \Cref{result:proof-of-maximal-min-entropy}. 

\begin{result}[Informal]\label[result]{result:proof-of-maximal-conditional-min-entropy}
    We define a notion of soundness that guarantees maximal conditional min-entropy, and we construct a proof of \emph{conditional} min-entropy in the QROM that satisfies this soundness property.
\end{result}

The soundness property says that the conditional min-entropy of $X$ given $C$ is maximal even when $C$ and $\Pi$ were chosen adversarially. We consider the notion of average conditional min-entropy $H_\infty^\mathsf{avg}$ from \cite{DORS03}. More formally, soundness says that for any malicious polynomial-time prover who chooses $\Pi$ and $C$,
\[\Pr_H\left[\Pr[X \neq \bot] \geq \frac{1}{\poly(\secp)} \land H_\infty^\mathsf{avg}[X | C, X \neq \bot] \leq n - \omega(\log \secp)\right] = \negl(\secp)\]

To achieve soundness, we again use \cite{CKMRST26}'s proof of min-entropy to generate a source $W$ with high min-entropy. Then we hash both $C$ and $W$ through a compressing random oracle $G$: i.e. $X = G(C, W)$. Using the measure-and-reprogram technique, we can show that $H_\infty^\mathsf{avg}[X | C, X \neq \bot]$ is maximal.

\subsection{Application to Randomness Beacons} \label{sec:entropy-beacon}
A randomness beacon periodically publishes fresh random pulses, together
with timestamps and authentication data. The pulses are used to make
randomized procedures publicly auditable \cite{Rabin83,NIST19}. A direct
application of our proof of maximal conditional min-entropy is a \emph{public randomness beacon with maximal min-entropy}. Our beacon can be generated on a single untrusted quantum device, and our security proof is unconditional in the QROM. 

Here is how the randomness beacon would work. In epoch
$t$, define the label by a canonical, domain-separated encoding
\[
  c_t=\langle\textsf{beacon-id},t,
        P_1,\ldots,P_{t-1},M_t\rangle,
\]
where $P_1,\ldots,P_{t-1}$ are the preceding pulses and $M_t$ contains any
public metadata or application commitment that must be bound to the new
pulse.  A quantum beacon operator generates a proof $\pi_t \gets \mathsf{MaxEnt.Prove}^{H}(1^\lambda,1^n,c_t)$ and publishes it. Then every user
deterministically computes
\[
  x_t=\mathsf{MaxEnt.Verify}^{H}
       (1^\lambda,1^n,c_t,\pi_t).
\]
The pulse contains $(c_t,\pi_t,x_t)$. The application may add signatures, consensus, timestamps, and
hash chaining to authenticate and order the pulses, but they are not used
in the entropy proof.

The label of a later pulse necessarily depends on earlier pulses and may also
depend on an adversarially influenced public transcript.  This is why the guarantee of conditional min-entropy will be important for security
(\Cref{def:adaptive,thm:adaptive-maximal-min-entropy}). We can treat the
entire history generation procedure as part of the adversarially chosen label $C_t$, and the proof certifies min-entropy even conditioned on the resulting $C_t$. For every round whose
acceptance probability is at least an inverse polynomial, the accepted pulse
satisfies
\[
  \Havg(X_t\mid C_t,\Acc_t,H)>n-\gamma
\]
except for a negligible fraction of random oracles. 

This statement gives two useful beacon properties.  First, before the proof
is published, a predictor that is given the epoch label and the public random
oracle still only has probability
$\leq 2^{-n+\gamma}$ of guessing the next pulse.  Second, if an adversary's attack depends on the condition that the pulse lies in some bad set $B$, then
then we can bound the success of this attack via a union bound combined with the min-entropy guarantee.  Thus, we recover security notions comparable to a true random beacon up to an unavoidable logarithmic
entropy loss. 

The closest notion of security in the literature seems to be $g$-optativity in~\cite{KMQR21}, where $g$ is a parameter and $g$-optativity implies a min-entropy deficiency of $\log g$. Their protocol requires a trusted seed whereas ours is seedless, and our protocol can be run on a single untrusted quantum device.
We note, however, that our protocol does not provide liveness, and it gives no entropy guarantee
if the operator publishes a rejecting proof with overwhelming probability.  A deployed
beacon therefore needs a separate liveness mechanism, such as an honest
fallback prover or a protocol for choosing the first timely valid pulse.

Finally, Kavuri et al.~\cite{KPR+25} demonstrated a continuously operating,
device\hyp{}independent quantum randomness beacon based on space-like separated
Bell measurements, which relies on a physical no-communication assumption. But their rate is subconstant and they do not satisfy our condition min-entropy property.

\subsection{Related Work}
\paragraph{Certified randomness.}

\begin{table}[ht]
\centering
\resizebox{\textwidth}{!}{%
\begin{tabular}{@{}llccccc@{}}
\toprule
  & \textbf{Assumption/Model}
  & \textbf{NI}
  & \textbf{Length ($n$)}
  & \textbf{Min-Entropy ($h$)}
  & \textbf{Rate ($h/n$)} \\
\midrule
\cite{AH23}
  & LLH; exp.\ verifier
  & \textcolor{BrickRed}{\textbf{no}}
  & $\Theta(\lambda^3\log \lambda)$
  & $\Omega(\lambda \log \lambda)$
  & \textcolor{BrickRed}{$\boldsymbol{\Omega(1/\lambda^2)}$} \\
\cite{YZ24}
  & AA; QROM
  & \textcolor{ForestGreen}{\textbf{yes}}
  & $\Theta(\lambda^2 \log \lambda)$
  & $o(\lambda)$
  & \textcolor{BrickRed}{
    $\boldsymbol{o(1/\lambda \log \lambda)}$} \\
\cite{KRT26}
  & QROM; shallow
  & \textcolor{ForestGreen}{\textbf{yes}}
  & $\Theta(\lambda^2\log\lambda)$
  & $o(\lambda^{1/2})$
  & \textcolor{BrickRed}{
      $\boldsymbol{o(1/\lambda^{3/2}\log\lambda)}$} \\
\cite{CKMRST26}
  & QROM
  & \textcolor{ForestGreen}{\textbf{yes}}
  & $\Theta(\lambda^2\log\lambda)$
  & $o(\lambda^{1/2})$
  & \textcolor{BrickRed}{
      $\boldsymbol{o(1/\lambda^{3/2}\log\lambda)}$} \\
\cite{CNS26}
  & OSS, exp. LWE
  & \textcolor{BrickRed}{\textbf{no}}
  & $|\mathsf{pk}|=\Theta(\lambda^2)$
  & $\lambda/8$
  & \textcolor{BrickRed}{
    $\boldsymbol{\Omega(1/\lambda)}$} \\
\midrule
{\bf This work}
  & QROM
  & \textcolor{ForestGreen}{\textbf{yes}}
  & $n$
  & $n-\omega(\log \secp)$
  & \textcolor{ForestGreen}{$\boldsymbol{\sim 1}$}\\
\bottomrule
\end{tabular}%
}
\captionsetup{
  labelfont=bf,
  justification=raggedright,
  singlelinecheck=false
}
\caption{Min-Entropy Rates of Various Certified Randomness Protocols. We consider variants that do not use seeded extractors. \textcolor{ForestGreen}{\textbf{Green text}} indicates a stronger
feature and \textcolor{BrickRed}{\textbf{red text}} a limitation.  ``NI'' means
non-interactive; ``AA'' denotes the Aaronson--Ambainis conjecture; and ``LLH'' denotes long-list hardness.  The \cite{CNS26} output length uses the concrete
one-shot-signature construction of \cite{HSVZ26}.}
\label{tab:min-entropy-comparison}
\end{table}

Certified randomness protocols allow a classical verifier with limited randomness to check that a quantum prover generated a high-entropy string. Certified randomness is slated to be one of the first applications of quantum computers in various domains. \cite{ACC+25} discusses various applications of certified randomness.

The device-independent line of work certifies randomness from violations of
Bell inequalities, under the assumption that physically separated quantum devices do not communicate during the protocol \cite{PironioEtAl10,VV12,MS16,MS17}.
These protocols first lower-bound the smooth min-entropy accumulated over the
measurement transcript and then apply a quantum-proof strong extractor.  

The same division between certification and extraction appears in 
single-device protocols.  Brakerski, Christiano, Mahadev, Vazirani, and
Vidick constructed an interactive protocol from noisy trapdoor claw-free
functions, instantiable from LWE \cite{BCMVV21}.  Their analysis accumulates
linear smooth min-entropy over many rounds and finishes with a quantum-proof
strong extractor using only polylogarithmically many additional uniform bits.
Mahadev, Vazirani, and Vidick subsequently obtained $\Omega(n)$ smooth
min-entropy in constant rounds with $\widetilde O(n)$ communication by using
lossy LWE matrices \cite{MVV22}.  These LWE-based protocols have efficient
classical verification, but verification uses a secret trapdoor and the
protocol is interactive.

Aaronson and Hung proposed a different single-device route based on random
circuit sampling \cite{AH23}.  Samples that pass a linear cross-entropy
benchmark contain $\Omega(n)$ min-entropy under their long-list hardness
assumption. They also give random-oracle evidence for the assumption and an
analysis with entangled side information in the oracle setting.  They
accumulate entropy across many circuits and apply a quantum-secure seeded
extractor---concretely, a Trevisan extractor---to obtain nearly uniform bits.
Pseudorandom challenge circuits reduce the initial random seed and yield net
randomness expansion.  The main tradeoff is that classical verification of
the benchmark takes exponential time.  This approach has also motivated an
experimental demonstration on a trapped-ion processor \cite{LiuEtAl25}.

Yamakawa and Zhandry introduced proofs of min-entropy in the QROM and showed,
assuming the Aaronson--Ambainis (AA) conjecture, that their non-interactive,
publicly verifiable proof of quantumness certifies min-entropy with rate approaching $1/\lambda \log\lambda $~\cite{YZ24}.  To convert a proof of
min-entropy to nearly uniform randomness, they also apply a strong extractor: the verifier samples a private uniform seed after receiving the prover's
message and may reveal the seed afterward.  This last seed adds a verifier
message, so the proof of \emph{randomness} is no longer seedless or
non-interactive. 

Two recent works removed the AA conjecture to certify min-entropy in the QROM unconditionally. Khurana, Roberts, and Tal proved soundness of the
original Yamakawa--Zhandry construction against adversaries making their
queries in $o(\log\lambda)$ adaptive layers \cite{KRT26}. Coladangelo,
Khurana, Mutreja, Roberts, Slote, and Tal proved soundness of a modified protocol against adversaries making $2^{o(\secp^c)}$ fully adaptive quantum
queries \cite{CKMRST26}. These protocols remain single-prover,
non-interactive, and publicly verifiable.  

In the standard model, Casper,
Nehoran, and Sattath constructed publicly certifiable min-entropy from
one-shot signatures \cite{CNS26}.  Their basic construction certifies
superlogarithmic min-entropy; an exponentially secure variant certifies
$\lambda/8$ bits.  It supports classical communication and transferability
with a quantum verifier, or classical verification without transferability,
under stronger assumptions.  It does not provide a maximal entropy rate.

\paragraph{Random oracles as seedless extractors.}
Several prior works~\cite{CDKT19,DVW19} study the use of random oracles as extractors. Coretti, Dodis, Karthikeyan, and Tessaro~\cite{CDKT19} show that a random
oracle can act as a seedless extractor even for oracle-dependent sources.
But their guarantee applies only to \emph{legitimate} sources: the
source must retain high min-entropy even conditioned on the sampler's state
and its complete oracle-query transcript. Dodis, Vaikuntanathan, and
Wichs~\cite{DVW19} study the related problem of extractor-dependent sources
for seeded, standard-model extractors. Their definition additionally requires
that the source query the extractor at its eventual output only with
negligible probability.

As we discussed above, it is impossible to guarantee that the output is uniformly random if the source is generated by a malicious adversary that queries the random oracle. We address this problem one way: by guaranteeing slightly lower min-entropy, $n - \omega(\log \secp)$, without making any restrictions on the adversary. These other works address the same problem in a different way: by restricting the adversary's query behavior in order to achieve uniform randomness. These restrictions on the adversary rule out the standard rejection sampling attack and guarantee that the extractors'
output is indistinguishable from uniform. However the restrictions are not typically enforceable against malicious adversaries.

\paragraph{Other seedless extractors.}
Other works~\cite{TrevisanVadhan,BGDM23,OhShaltiel26} have constructed single-source extractors for samplable sources. However the approach of\ifSubmission~\cite{TrevisanVadhan} and~\cite{BGDM23}\else~\cite{TrevisanVadhan,BGDM23}\fi~requires knowing the runtime of the adversary ahead of time, and requires the verifier to do computation longer than the adversary. This is not compatible with our cryptographic setting. Further, the recent approach of \cite{OhShaltiel26} can currently only produce a single-bit output, as it uses the two-source extractor recipe where generalizing to arbitrary bits is an open problem.

\paragraph{Random oracles as seedless condensers.}
Condensers, like extractors, compress a random string to increase the min-entropy rate, although they do not necessarily produce uniform randomness. Prior work has constructed seedless condensers for efficiently sampleable sources. The state-of-the-art~\cite{DRV12,FSW25} achieves the maximal rate (with $O(\log \secp)$ entropy loss) under the multi-collision resistance of a random function. They require an $\ell$-collision-resistant hash function, where $\ell$ grows polynomially with the security parameter. Combining this with work on multi-collision-resistance in the QROM~\cite{YZ20} may provide an alternate path to proving the maximal min-entropy guarantee, although it would not achieve certified \emph{conditional} min-entropy.
\section{Technical Overview}

\subsection{Definition}
A proof of min-entropy allows a quantum prover to sample a high min-entropy string and convince a verifier that they have done so. In the variant from \cite{YZ24,KRT26,CKMRST26}, the proof of min-entropy comprises the algorithms $(\Prove, \Verify)$ with the following syntax.

\begin{itemize}
    \item $\Prove^{H}(1^\lambda,1^h) \to \Pi:$ A QPT algorithm that samples a classical proof $\Pi$ with min-entropy $\geq h$.
    \item $\Verify^{H}(1^{\lambda}, 1^h, \Pi) \to X:$ A \emph{deterministic} classical polynomial time algorithm that verifies $\Pi$ and uses it to derive a string $X$ with certified min-entropy $\geq h$. If $\Pi$ is rejected, then $\Verify$ sets $X = \bot$.
\end{itemize}
 
We work in the quantum random oracle model, so $\Prove$, $\Verify$, and the adversary can make a polynomial number of queries to a public random oracle $H$. Unlike \cite{CDKT19,DVW19}, we make no restriction on the types of queries that the adversary makes to the random oracle. We only restrict the number of queries to be polynomial.

Soundness says roughly that if $\Pi$ is generated by an adversarial prover, they cannot fool the verifier into accepting a low-min-entropy string $X$. Given any adversarial prover, this prover defines the distribution of $\Pi$ and $X$. We require that this adversary either fails verification $(X = \bot)$ with overwhelming probability or the min-entropy of $X$, conditioned on $X \neq \bot$, is higher than $h$. More formally, soundness says that for any poly-query adversary $\calA$ that generates $\Pi$,
\[\Pr_{H}\Big[\Pr[X \neq \bot]\geq \frac{1}{\poly(\secp)} \wedge H_{\infty}\Big( X | X \neq \bot\Big)\leq h(\lambda)\Big]\leq \negl(\lambda)\]
Note that $\Pr[X \neq \bot]$ and $H_{\infty}\Big( X | X \neq \bot\Big)$ implicitly condition on the value of $\tilde{h} \gets H$ because they are inside of $\Pr_{H}[\cdot]$. This means that $X$ has min-entropy $> h$, even conditioned on the randomness of $H$, so the protocol is generating new randomness beyond the randomness supplied by the random oracle.

In this work, we strengthen the proof of min-entropy in two ways. First, we enable a proof of conditional min-entropy. Now $\Prove$ and $\Verify$ take as input a label $C$ and generate a string $X$ such that $X|C$ has high conditional min-entropy. This holds even if the adversary chooses $C$ and $\Pi$. Second, our protocol achieves maximal min-entropy. As long as the bitlength of $X$ is $n = \omega(\log \secp)$, the conditional min-entropy of $X|C$ is $\geq n - \gamma(\secp)$ for all $\gamma(\secp) = \omega(\log \secp)$. More formally, soundness says that for any poly-query adversary $\calA$ that generates $C$ and $\Pi$,
\[\Pr_{H}\Big[\Pr[X \neq \bot]\geq \frac{1}{\poly(\secp)} \wedge H_{\infty}^\mathsf{avg}\Big( X | C, X \neq \bot\Big)\leq n - \gamma(\lambda)\Big] \leq \negl(\lambda)\]
for all $\gamma(\secp) = \omega(\log \secp)$.
We call this primitive a \emph{proof of maximal conditional min-entropy}.

\subsection{Construction}\label[section]{sec:tech-overview-construction}
Now we will describe our construction of a proof of maximal conditional min-entropy (\Cref{constr:maxent}) and provide intuition for its soundness. We start by running \cite{CKMRST26}'s proof of min-entropy $\mathsf{MinEnt}$ using oracle $H$ to generate a weakly random string $W$ with min-entropy $\geq 4n$. Then we hash $(C, W)$ through a compressing random oracle $G$ to produce a shorter random string $X$ that has maximal min-entropy even conditioned on $C$. In summary,
\begin{itemize}
    \item $\mathsf{Prove}^{G,H}(1^{\lambda},1^n,C)$: Run $\Pi\leftarrow \mathsf{MinEnt.Prove}^{H}(1^{\lambda^{\star}},1^h)$, for appropriate parameters $\secp^\star$ and $h \geq 4n$, and output $\Pi$.

    \item  $\mathsf{Verify}^{G,H}(1^{\lambda},1^n,C,\Pi):$ 
    \begin{enumerate}
        \item Compute $W \leftarrow \mathsf{MinEnt.Verify}^{H}(1^{\lambda^{\star}},1^h,\Pi)$. If $W=\bot$, then output $X = \bot$. Otherwise continue.
        \item Output $X = G(C,W)$.
    \end{enumerate}
\end{itemize}
This construction uses two oracles $(H, G)$ for notational convenience, but they can be implemented using a single random oracle with a larger domain.

Here is some intuition for why $X = G(C, W)$ has maximal min-entropy even conditioned on $C$. First, random oracles are good seedless randomness condensers. We can view $G(C, \cdot)$ as a random oracle mapping $W \to X$, and it condenses a weakly random source $W$ to a string $X$ with maximal min-entropy. \cite{FSW25} proved that against \emph{classical} adversaries, random oracles work as seedless condensers and achieve maximal min-entropy. One contribution of our work is to show that this result holds against \emph{quantum} adversaries as well.

Second, each value of $C$ defines an independently random oracle $G(C, \cdot)$. Since $C$ is an input to $G$, we can imagine that the adversary decides on $C$ before obtaining any information about $G(C, \cdot)$. While this oversimplifies the adversary, it provides useful intuition. We make a similar but rigorous statement using the measure-and-reprogram technique \cite{DFM20}. If the adversary chooses $C$ before seeing $G(C, \cdot)$, then due to the min-entropy of $W$ and the condensing property of $G(C, \cdot)$, $X | C$ will have maximal conditional min-entropy.

\paragraph{Why existing extractors fall short.}
Prior work on extractors does not guarantee security against the fully malicious quantum adversaries that we consider here. For instance, our adversary can make quantum queries to $G$ before choosing the source and label $(W, C)$. Thus the truth table of $G$ is not an independent extractor seed, and the pair $(C,W)$ may be correlated with all
the oracle information obtained by the adversary. Neither an ordinary
seeded-extractor theorem nor a multi-source extractor applies to this
distribution. 

Prior work on random oracles as seedless extractors allows
some dependence between the source and the ideal primitive
\cite{CDKT19,DVW19}, but imposes freshness or ``legitimacy'' conditions on the
source's classical oracle queries, or computational restrictions on the
source and distinguisher.  Those models do not cover a source selected after
adaptive quantum queries to the same oracle. 

Our conclusion is deliberately weaker than uniform randomness. We certify
$n-\gamma$ bits of average conditional min-entropy in an $n$-bit output for
every $\gamma=\omega(\log\lambda)$.  This near-maximal guarantee is compatible
with the unavoidable polynomial-time rejection-sampling attack, which can
fix $O(\log\lambda)$ output bits.

\subsection{Soundness Proof}\label[section]{sec:tech-overview-soundness-proof}
Here we explain the techniques of our soundness proof. We will begin by placing several constraints on the adversary (\Cref{sec:tech-overview-light-queries-and-fixed-label}), such as making only light queries to $G$ and outputting a deterministic label $C = \mathbf{0}$. Then we will remove two of the constraints in \Cref{sec:tech-overview-heavy-queries-fixed-set} - \Cref{sec:tech-overview-adaptive-labels}.

\subsubsection{Case 1: Light Queries}\label[section]{sec:tech-overview-light-queries-and-fixed-label}

\paragraph{Simplifications.} To simplify the presentation in this section, let us assume an ideal adversary with the following constraints:
\begin{itemize}
    \item \textbf{Proof always accepts.} $\Pr[W \neq \bot] = 1$.
    \item \textbf{Min-entropy is high on every $H, G$.} $\Pr_{H, G}\Big[H_{\infty}(W)> h\Big] = 1$.
    \item \textbf{Light queries.} On any oracle $G$, any query $j \in [q]$, and any value $w \in \bit^m$, the query weight is $\wt_w^j(G) < \frac{2^{-3n}}{q^2}$.
    \item \textbf{Fixed label.} The adversary always outputs $C = \mathbf{0}$.
\end{itemize}

To prove soundness in this setting, it suffices to prove that for any possible output value $x \in \bit^n$ and any adversary making $\leq q = \poly(\secp)$ queries to $G$,
\[\Pr_{g \gets G}\left[\Pr[X=x | G = g] \geq 2^{-n + O(\log \secp)}\right] \leq 2^{-2^{\Omega(n)}}\]
This says that with high probability over $g \gets G$, $\Pr[X=x | G = g]$ is close to $2^{-n}$.

Then by taking a union bound over all possible $x \in \bit^n$, we have that with high probability over $g \gets G$, $\Pr[X=x | G = g]$ is close to $2^{-n}$ for all $x \in \bit^n$.
\begin{align*}
    \Pr_{g \gets G}\left[\exists x \in \bit^n, \Pr[X=x | G = g] \geq 2^{-n + O(\log \secp)}\right] &\leq 2^n \cdot 2^{-2^{\Omega(n)}} = \negl(\secp)
\end{align*}
We used the fact that $n = \omega(\log \secp)$. Then,
\begin{align*}
    \negl(\secp) &= \Pr_{g \gets G}\left[H_\infty(X | G = g) \leq n - O(\log \secp)\right] \ifSubmission{\\&}\fi= \Pr_{g \gets G}\left[H_\infty^{\mathsf{avg}}(X | C, G = g) \leq n - O(\log \secp)\right]
\end{align*}
We used the fact that $H_\infty^{\mathsf{avg}}(X | C, G = g) = H_\infty(X | G = g)$ because the label is fixed to $C = \mathbf{0}$.

\paragraph{Proof Intuition.} Here is some intuition for the rest of the security proof. Since each input $w$ is lightly queried, $\calA$ is very uncertain about whether $w \in G^{-1}(x)$ (the set of preimages of $x$). If $\calA$ outputs a value $w$ that happens to be in $G^{-1}(x)$, then this happened by dumb luck. To formalize this intuition, we give the adversary a reprogrammed oracle $G'$ that differs from the real oracle $G$ on a small number of inputs. The swapping lemma (\Cref{lem:swap}) and the \emph{light queries} constraint imply that the adversary's behavior will not change significantly from the reprogramming. However, $G'$ is sampled so that it gives no information about $G^{-1}(x)$. Then the final output $W$ is independent of $G^{-1}(x)$. In expectation over $G$, $W$ gives probability mass $2^{-n}$ to $G^{-1}(x)$. Furthermore, the Hoeffding bound shows this probability mass is tightly concentrated around its expected value. This will prove that with high probability over $G$, $W$ gives probability mass $\leq 2^{-n+O(\log \secp)}$ to $G^{-1}(x)$. We will make this argument formal below.

\paragraph{Sampling a reprogrammed oracle.} First, let us sample the reprogrammed oracle $G'$ along with $G$ so that $G'$ reveals no information about $G^{-1}(x)$. The procedure first samples all outputs of $G$ and $G'$ other than $x$ (determined by $E_0$), and second chooses $G^{-1}(x)$ and $G'^{-1}(x)$ (determined by $D$ and $D'$ respectively). The sampling procedure (based on \cref{def:completion}) is as follows:
\begin{enumerate}
    \item Sample a random oracle $E_0: \bit^m \to \bit^n$ subject to the constraint that $E_0$ never outputs $x$. That is to say, for each $w \in \bit^m$, sample $E_0(w) \getsr \bit^n \backslash \{x\}$.
    \item For each $w \in \bit^m$ sample two independent Bernoulli variables $D(w)$ and $D'(w)$ such that $\Pr[D(w) = 1] = \Pr[D'(w) = 1] = 2^{-n}$. 
    \item Define $G$ and $G'$:
    \[G(w) := \begin{cases}
                x,&D(w)=1,\\E_0(w),&\text{Otherwise}
            \end{cases}
            \qquad
    G'(w) := \begin{cases}
                x,&D'(w)=1,\\E_0(w),&\text{Otherwise}
            \end{cases}\]
    $D(w) = 1$ if and only if $w \in G^{-1}(x)$, and likewise for $D'$ and $G'^{-1}(x)$.
\end{enumerate}
Note that the marginal distribution of $G$ is uniformly random over all oracles mapping $\bit^m \to \bit^n$ (\Cref{lem:uniform}). Furthermore $G'$ reveals no information about $G^{-1}(x)$ because $G'$ is independent of $D$.

Second, we can express $\Pr[X=x | G = g]$ in terms of quantum states and projectors. For any $(g, g', d, d') \in \mathsf{support}(G, G', D, D')$, let $|\psi^{g}\rangle$ be the superposition over $W$-values resulting from the adversary with query access to oracle $g$. $|\psi^{g'}\rangle$ is defined analogously. Let us define the projector $\Pi_{g,x}$ that projects $|\psi^{g}\rangle$ or $|\psi^{g'}\rangle$ onto all preimages of $x$ under $g$:
\[\Pi_{g,x} := \sum_{w\in\bits^m} d(w) \ketbra{w}{w}\]
Then,
\begin{align*}
    \Pr[X = x | G = g] &= \sum_{w\in\bits^m} \mathbbm{1}_{g(w)=x} \cdot \abs{\braket{w}{\psi^g}}^2 = \norm{\Pi_{g,x} |\psi^{g}\rangle}^2
\end{align*}
Next, instead of giving the adversary query access to $g$, we will give them query access to $g'$ so that they will not know $G^{-1}(x)$. If $|\psi^{g}\rangle$ and $|\psi^{g'}\rangle$ are close in Euclidean distance, then this will not significantly decrease the adversary's probability mass on $G^{-1}(x)$. Using standard tools, we can show that
\[\Pr[X = x | G = g] \leq 2\norm{\Pi_{g,x} |\psi^{g'}\rangle}^2 + 2\norm{|\psi^{g}\rangle-|\psi^{g'}\rangle}^{2}\]

It just remains to show that except with probability $2^{-2^{\Omega(n)}}$, $\norm{\Pi_{g,x} |\psi^{g'}\rangle}^2$ and $\norm{|\psi^{g}\rangle-|\psi^{g'}\rangle}^{2}$ are $< 2^{-n + O(\log \secp)}$. We will do this in \Cref{lem:tech-overview-bound-on-success-probability} and \Cref{lem:tech-overview-psi-G-psi-g-prime-close}, respectively.

Third, we show that $\norm{\Pi_{G,x} |\psi^{G'}\rangle}^2 \approx 2^{-n}$ with high probability. We condition on the value of $G'$ to express $\norm{\Pi_{G,x} |\psi^{G'}\rangle}^2$ as a sum of independent random variables, which we then bound with the Hoeffding bound. 
\begin{lemma}[Informal]\label[lemma]{lem:tech-overview-bound-on-success-probability}
    For any $x \in \mathsf{support}(X)$,
    \[\Pr_{(g, g') \gets (G, G')}\left[\norm{\Pi_{g,x} |\psi^{g'}\rangle}^2 \geq 2^{-n + 1}\right] \leq 2^{-2^{\Omega(n)}}\]
\end{lemma}
\begin{proof}[Proof Sketch]
First, note that \Cref{lem:tech-overview-bound-on-success-probability} is equivalent to the following claim:
\begin{align*}
    \Pr_{G, G'}\left[\sum_{w\in\bits^m} D(w) \cdot \abs{\braket{w}{\psi^{G'}}}^2 \geq 2^{-n + 1}\right] &\leq 2^{-2^{\Omega(n)}}
\end{align*}

Second, we will prove the claim above using the Hoeffding bound. For each $w \in \bit^m$, define the variable $Z_w = D(w) \cdot \abs{\braket{w}{\psi^{G'}}}^2$, and let us condition on the value of $g' \gets G'$. Note that conditioning on $G'$ does not change the distribution of $D$ because $G'$ and $D$ are independent. 

The Hoeffding bound (\Cref{thm:hoeffding-inequality}) says that if $\left[Z_w\right]_{w \in \bit^m}$ are independent and bounded ($a_w \leq Z_w \leq b_w$ for all $w$), then
\[\Pr\left[\sum_{w \in \bit^m}Z_w \geq \bbE\left[\sum_{w \in \bit^m}Z_w\right] + t\right] \leq e^{-\frac{2 t^2}{\sum_{w} (b_w - a_w)^2}}\]

Note that each $Z_w$ is bounded: if $a_w = 0$ and $b_w = \abs{\braket{w}{\psi^{G'}}}^2$, then $a_w \leq Z_w \leq b_w$. Furthermore, after conditioning on the value of $g' \gets G'$, the variables $\left[Z_w\right]_{w \in \bit^m}$ are independent because the remaining randomness of each $Z_w$ comes from $D(w)$.

Third, let us bound $\sum_{w} (b_w - a_w)^2$ using the min-entropy of $W$. Since $H_\infty(W) > h \geq 4n$, we have that $\abs{\braket{w}{\psi^{g'}}}^2 = \Pr[W = w] \leq 2^{-4n}$. This implies that:
\begin{align*}
    \sum_{w \in \bit^m} (b_w - a_w)^2 &= \sum_{w \in \bit^m} \abs{\braket{w}{\psi^{g'}}}^4 \leq \sum_{w \in \bit^m} \abs{\braket{w}{\psi^{g'}}}^2 \cdot 2^{-4n} = 2^{-4n}
\end{align*}

Fourth, let us set $t = 2^{-n}$. Note that 
\begin{align*}
    \bbE\left[\sum_{w \in \bit^m} Z_w\right] &= \sum_{w \in \bit^m} \Pr[G(w) = x] \cdot \abs{\braket{w}{\psi^{g'}}}^2 \ifSubmission{\\&}\fi= 2^{-n} \cdot \sum_{w \in \bit^m} \abs{\braket{w}{\psi^{g'}}}^2 = 2^{-n}
\end{align*}

Therefore, $\bbE\left[\sum_{w \in \bit^m} Z_w\right] + t = 2^{-n} + 2^{-n} = 2^{-n+1}$.

Finally, the Hoeffding bound (\Cref{thm:hoeffding-inequality}) implies that:
\begin{align*}
    \Pr_{G}\left[\sum_{w\in\bits^m} D(w) \cdot \abs{\braket{w}{\psi^{g'}}}^2 \geq 2^{-n + 1}\right] &\leq e^{-\frac{2 \cdot 2^{-2n}}{2^{-4n}}} = e^{-2^{2n+1}} = 2^{-2^{\Omega(n)}}
\end{align*}
Taking the expectation over $G'$ completes the proof.
\end{proof}

Finally, we will show that $|\psi^{G}\rangle$ and $|\psi^{G'}\rangle$ are close, with high probability. Crucially, $E_0$ will serve as a dependency-breaking variable, allowing us to bound $\left\||\psi^{G}\rangle - |\psi^{G'}\rangle\right\|_2^2$ with a sum of \emph{independent} random variables, which can in turn be bounded with the Hoeffding bound. 
\begin{lemma}[Informal]\label[lemma]{lem:tech-overview-psi-G-psi-g-prime-close}
    \[\Prob_{G, G'}\!\left[\norm{|\psi^{G}\rangle-|\psi^{G'}\rangle}^{2} \geq 2^{-n+O(\log \secp)} \right]  \leq 2^{-2^{\Omega(n)}}\]
\end{lemma}
\begin{proof}[Proof Sketch]

First, let $\ket{\phi} = \ket{\psi^{E_0}}$. This is a hybrid state generated without knowledge of $D$ or $D'$. Since $G$ differs from $E_0$ only when $D(w) = 1$, the swapping lemma (\cref{thm:swapping-lemma}) implies that
\begin{align*}
    \left\||\phi\rangle - |\psi^{G}\rangle\right\|^2 \leq 4 q \cdot \sum_{w \in \bit^m} D(w) \cdot \wt_{w}^{\leq q}(E_0)
\end{align*}
where $\wt_{w}^{\leq q}(E_0)$ is the cumulative query weight given to $w$ when the adversary has query access to $E_0$. The \emph{light queries} constraint implies that $\wt_{w}^{\leq q}(E_0) < q \cdot \frac{2^{-3n}}{q^2} = \frac{2^{-3n}}{q}$.

Now let us condition on the value of $E_0$. Then the terms $Z_w := D(w) \cdot 4 q \wt_{w}^{\leq q}(E_0)$ are independent random variables because the randomness of each $Z_w$ comes only from $D(w)$. Then we can apply the Hoeffding bound to the equation above.

Second, each $Z_w$ is bounded. If $a_w = 0$ and $b_w = 4 q \wt_{w}^{\leq q}(E_0)$, then $a_w \leq Z_w \leq b_w$. Additionally,
\begin{align*}
    \sum_{w \in \bit^m} (b_w - a_w)^2 &= \sum_{w \in \bit^m} \left(4 q \wt_{w}^{\leq q}(E_0)\right)^2\\
    &\leq \sum_{w \in \bit^m} 16 q^2 \cdot \frac{2^{-3n}}{q} \cdot \wt_{w}^{\leq q}(E_0) \leq 2^{-3n+O(\log \secp)}.
\end{align*}

Third, $\bbE\left[\sum_{w \in \bit^m} Z_w\right] \approx 2^{-n}$. Conditioned on the value of $E_0$,
\[\bbE\left[\sum_{w \in \bit^m} Z_w\right] = \sum_{w \in \bit^m} 2^{-n} \cdot 4 q \wt_{w}^{\leq q}(E_0) = 2^{-n} \cdot 4 q^2 = 2^{-n + O(\log \secp)}\]

Fourth, by applying the Hoeffding bound, we can show that
\[\Prob_{G}\!\left[\norm{|\phi\rangle-|\psi^{G}\rangle}^{2} \geq 2^{-n+O(\log \secp)} \right] \leq e^{-2^{\Omega(n)}}\]

Fifth, by the same reasoning, we can show that $\ket{\phi}$ is close to $|\psi^{G'}\rangle$:
\[\Prob_{G'}\!\left[\norm{|\phi\rangle-|\psi^{G'}\rangle}^{2} \geq 2^{-n+O(\log \secp)} \right] \leq e^{-2^{\Omega(n)}}\]
Then by the triangle inequality,
\[\Prob_{G, G'}\!\left[\norm{|\psi^{G}\rangle-|\psi^{G'}\rangle}^{2} \geq 2^{-n+O(\log \secp)} \right] \leq 2^{-2^{\Omega(n)}}\]
\end{proof}

\subsubsection{Case 2: Heavy queries on a fixed set}\label[section]{sec:tech-overview-heavy-queries-fixed-set}
Here we handle adversaries that don't follow the \emph{light queries} constraint. These adversaries may give heavy query weight ($\geq \frac{2^{-3n}}{q^2}$) to some inputs. However, to make things simple, let us assume that the set of heavily queried points is the same for all oracles $G$. 

For a given oracle $G$, let the heavily queried set $\Delta$ be the set of all values $w \in \bit^m$ for which the query weight is large on some query $j \in [q]$: $\wt_w^j(G) \geq \frac{2^{-3n}}{q^2}$. Then, let us remove the \emph{light queries} constraint, and instead assume the following weaker constraint:
\begin{itemize}
    \item \textbf{Heavy queries on a fixed set.} All oracles $G$ have the same heavily queried set.
\end{itemize}

Unlike in \Cref{sec:tech-overview-light-queries-and-fixed-label}, we cannot reprogram $G$ on any given $w$-value because if we reprogram on a value in $\Delta$, then $|\psi^{G'}\rangle$ might be far from $|\psi^{G}\rangle$. Our solution is to fix $G(w) = G'(w)$ on all heavily queried inputs $w \in \Delta$ and to sample $G^{-1}(x)$ and $G'^{-1}(x)$ independently on the lightly queried inputs.

\paragraph{Proof Intuition.} On the heavily queried inputs, the adversary knows with high certainty which of them belong to $G^{-1}(x)$. The adversary might choose among $\Delta$ to output only the ones in $G^{-1}(x)$. However, their knowledge of $G$ evaluated on $\Delta$ will not let them significantly increase $\Pr[X = x]$. The size of $\Delta$ is relatively small ($|\Delta| \leq q^3 2^{3n}$). Since $H_{\infty}(W) > 4n$, the total probability mass that $W$ gives to $\Delta$ is $\leq q^3 2^{3n} \cdot 2^{-4n} = q^3 2^{-n}$, which is small enough to ignore. 

To make this intuition formal, we will choose $G'$ to agree with $G$ on all inputs in $\Delta$ (including the ones in $G^{-1}(x)$). However, $G'$ gives no information about which lightly queried inputs belong to $G^{-1}(x)$. If the adversary is given query access to $G'$, then with high probability, their final probability mass on $G^{-1}(x)$ will be 
\begin{align*}
    \Pr[X=x] &= \Pr[X = x \land W \notin \Delta] + \Pr[X=x \land W \in \Delta]\\
    &\leq 2^{-n + O(\log \secp)} + q^3 2^{-n} = 2^{-n + O(\log \secp)}
\end{align*}

\subsubsection{Case 3: Heavy queries on an adaptive set}\label[section]{sec:tech-overview-heavy-queries-adaptive-set}
Here we remove all restrictions on the types of queries that the adversary can make to $G$. In \Cref{sec:tech-overview-heavy-queries-fixed-set}, the adversary's set $\Delta$ of heavily queried inputs is constant for all oracles $G$. In this section, the adversary may choose $\Delta$ adaptively based on $G$. For example, after making the first query, the adversary may use the response to decide which inputs to heavily query on the second query.

If the adversary chooses its query weight adaptively, then $\Delta$ itself may be correlated with $D(w)$ on the lightly queried inputs. Our goal is to sample $D, D', \Delta$ so that $D$ and $D'$ are independent on inputs outside of $\Delta$, yet $\norm{|\psi^{G}\rangle-|\psi^{G'}\rangle}$ is small.

Our solution is to sample $D, D', \Delta$ lazily. Our sampling procedure is similar to the one in \Cref{sec:tech-overview-light-queries-and-fixed-label} except that we wait to sample each $D(w)$ until $w$ is about to receive heavy query weight. At that point, we sample $D(w)$ and reprogram the oracle $E_0$ as necessary, to map $w \to x$ if $D(w) = 1$. The sampling procedure is as follows:
\begin{enumerate}
    \item Sample a random oracle $E_0: \bit^m \to \bit^n$ subject to the constraint that $E_0$ never outputs $x$. That is to say, for each $w \in \bit^m$, sample $E_0(w) \getsr \bit^n \backslash \{x\}$.
    \item On each query $j \in [q]$,
    \begin{enumerate}
        \item Let $\Delta_j$ be the set of inputs $w$ that receive large query weight ($\wt_w^j \geq \frac{2^{-3n}}{q^2}$) and do not belong to any earlier $\Delta_{j'}$ for any $j' < j$.
        \item For each $w \in \Delta_j$, sample $D(w) = D'(w) \gets \mathsf{Bernoulli}\left(2^{-n}\right)$. For any $w$ on which $D$ has not been sampled yet, let $D(w) = D'(w) = 0$ as a placeholder.
        \item Set $E_j(w) = \begin{cases}
            x, &D(w) = 1\\
            E_{0}(w), &\text{Otherwise}
        \end{cases}$
        \item Respond to the query using $E_j$.
    \end{enumerate}
    \item Let $\ket{\phi}$ be the adversary's final state, and let $\Delta = \Delta_1 \cup \ldots \cup \Delta_q$.
    \item For each remaining $w \notin \Delta$, sample two independent Bernoulli variables $D(w) \gets \mathsf{Bernoulli}\left(2^{-n}\right)$ and $D'(w) \gets \mathsf{Bernoulli}\left(2^{-n}\right)$.
    \item Define $G$ and $G'$:
    \[G(w) := \begin{cases}
                x,&D(w)=1,\\E_0(w),&\text{Otherwise},
            \end{cases}
            \qquad
    G'(w) := \begin{cases}
                x,&D'(w)=1,\\E_0(w),&\text{Otherwise}
            \end{cases}\]
\end{enumerate}
This procedure has the following nice properties. 

First, the reprogramming does not significantly change the adversary's behavior. We can show that with high probability,
\begin{align*}
    \norm{|\phi\rangle-|\psi^{G}\rangle}^2 &\leq 2^{-n + O(\log \secp)}\\
    \norm{|\phi\rangle-|\psi^{G'}\rangle}^2 &\leq 2^{-n + O(\log \secp)}
\end{align*}
The proof of this fact applies the swapping lemma after each query. During query $j$, right before the adversary puts heavy query weight on a value $w \in \Delta_j$, we reprogram the adversary's oracle $E_j$ so that $E_j(w) = G(w) = G'(w)$. Then the adversary never puts large query weight on an input where $E_j$ differs from $G$ or $G'$. Note that while we reprogram $E_j(w)$ during query $j$, the adversary may have given light query weight to $w$ on a prior query, before the reprogramming. On earlier queries $j' < j$, $E_{j'}(w)$ disagrees with $G(w)$ and $G'(w)$, but this does not significantly affect the adversary's behavior because their query weight on $w$ was small during query $j'$.

Second, $G'$ is independent of $D$, after conditioning on $E_0$, $\Delta$, and $D|_{\Delta}$. This is because $E_0, \Delta, D|_{\Delta}$ were sampled before we sampled $D(w)$ and $D'(w)$ for values of $w \notin \Delta$.

The two properties above allow us to complete the proof of security along similar lines as in \Cref{sec:tech-overview-light-queries-and-fixed-label}. The main difference is that we will condition on $(E_0, \Delta, D|_{\Delta})$ to break the dependence between $G'$ and $D$.

\subsubsection{Case 4: Adaptive Labels}\label[section]{sec:tech-overview-adaptive-labels}
In Cases 1-3, the adversary always outputs the same label $C = \mathbf{0}$, so we only
need to analyze one fixed slice $G(\mathbf{0}, \cdot)$ of the outer random oracle $G$. We now
remove the \emph{fixed label} constraint. The adversary may make adaptive quantum
queries to the entire labelled oracle $
G:\{0,1\}^{\ell}\times\{0,1\}^{m}\longrightarrow\{0,1\}^{n}
$
before choosing $C$ and $\Pi$. 

While the previous cases show that $H_\infty(X)$ is maximal, if $X$ and $C$ are correlated, the prior techniques cannot rule out the possibility that $H_\infty^\mathsf{avg}(X|C)$ is small. Our solution is to use the measure-and-reprogram technique~\cite{DFM20} to move to an ideal world where the adversary must choose their eventual output $C$ before they can make any queries to $G(C, \cdot)$. Once they start querying $G(C, \cdot)$, they cannot change the value of $C$, so they are in essentially the same position as the adversary in Case 3. The only way to increase $\Pr[X = x]$ is to choose $W$ to lie in $G(C, \cdot)^{-1}(x)$.

\paragraph{Making the selected slice fresh.}
We use the block measure-and-reprogram lemma (Lemma \ref{lem:block-mr}) to move to an ideal world where the adversary must choose $C$ before querying $G(C, \cdot)$.

The ideal-world adversary (the simulator) operates in two stages. In the first stage, the simulator queries an oracle  $G_0:\{0,1\}^{\ell}\times\{0,1\}^{m}\longrightarrow\{0,1\}^{n}$, then outputs a classical label $C$, and retains a residual quantum state. For the second stage, the simulator's oracle is reprogrammed on the $C$-th slice to a freshly random function $\Theta : \{0,1\}^{m}\longrightarrow\{0,1\}^{n}$. More formally, in the second stage, the simulator receives query access to the following oracle $G^\star$:
$$
G^\star(c',w)=
\begin{cases}
\Theta(w),&c'=C,\\
G_0(c',w),&c'\ne C.
\end{cases}
$$
Crucially, $\Theta$ is sampled independently of both $C$ and the residual state. 

\Cref{lem:block-mr} says that compared to the real-world adversary, the ideal-world simulator cannot do significantly worse at any given task. For any
specified success event, say $V$, the following holds:
$$
\Pr_{\mathrm{sim}}\bigl[V(G^\star,H,C,\Pi)\bigr]
\ge
\frac{1}{\operatorname{poly}(\lambda)}
\Pr_{\mathrm{real}}\bigl[V(G,H,C,\Pi)\bigr].
$$

\paragraph{Bounding the simulator's success.}
Now we can reduce the ideal world to Case 3 (\Cref{sec:tech-overview-heavy-queries-adaptive-set}) because the value of $G(C, \cdot) = \Theta(\cdot)$ is independent of $C$.

The remaining issue is the entropy after conditioning on the
label.  Section \ref{sec:lab} bounds the joint probability of a label and proof:
$$
\max_{\substack{c\in\{0,1\}^{\ell}\\w\in\{0,1\}^{m}}}
\Pr_{\mathrm{sim}}[C=c,W=w\mid\Theta=\theta]
\le 2^{1-h}.
$$
Whenever this bound holds, a label which is not ``rare'', or too unlikely, satisfies:
\begin{align*}
    \Pr[C = c] \ge 2^{4n+1-h}
    \quad\Longrightarrow\quad
    \max_{w\in\{0,1\}^{m}}
    \Pr_{\mathrm{sim}}[W=w\mid C=c,\Theta=\theta]
    &\le
    \frac{2^{1-h}}{\Pr[C = c]}\ifSubmission{\\&\quad\quad}\fi
    \le 2^{-4n}.
\end{align*}
This is precisely the point-probability bound used in Cases 1--3.
The heavy-set contribution is small as in Case 2, while Case 3's lazy
sampling handles adaptively chosen heavy sets and allows the swapping
and Hoeffding bounds to apply as before. 

For rare labels, that is labels $c$ such that $\Pr[C = c]<2^{4n+1-h}$,
the total probability mass is at most
$
\sum_{c:\Pr[C = c]<2^{4n+1-h}}\Pr[C = c]
\le 2^{\ell+4n+1-h}
$.
We show that this contribution is
small enough to absorb into the error term.
\section{Preliminaries}
\subsection{Probability}

\begin{definition}[Min-Entropy and Average Conditional Min-Entropy, \allowbreak\cite{DORS03} Section 2.4]\label[definition]{def:min-entropy}
    Let $A, B$ be random variables. We define the min-entropy $H_\infty$ and average conditional min-entropy $H_\infty^{\mathsf{avg}}$ as follows:
    \begin{align*}
        H_\infty(A) &= - \log \max_{a \in \mathsf{support}(A)} \Pr[A = a]\\
        H_\infty^{\mathsf{avg}}(A|B) &= - \log\left(\bbE_{b \gets B}\left[\max_{a \in \mathsf{support}(A)} \Pr[A = a | B = b]\right]\right)\\
        &= - \log\left(\sum_{b \in \mathsf{support}(B)} \max_{a \in \mathsf{support}(A)} \Pr[A = a, B = b]\right)
    \end{align*}
\end{definition}
Let us clarify some complicated notation that arises later on. Let $H, X, C$ be random variables, and let $E$ be an event determined by the values of $(H, X, C)$ such that for every value $\tilde{h} \in \mathsf{support}(H)$, $\Pr[E | H = \tilde{h}] > 0$. In the expression
\[\Pr_H\left[H_\infty(X|E) \geq h\right],\]
the min-entropy $H_\infty(X|E)$ implicitly conditions on the value of $H$ because it is inside of $\Pr_H[\cdot]$. Then the expression above is equivalent to the following:
\begin{align*}
    \Pr_H\left[H_\infty(X|E) \geq h\right] &= \sum_{\tilde{h} \in \mathsf{support}(H)}\Pr[H = \tilde{h}] \cdot \mathbbm{1}_{H_\infty(X|E, H = \tilde{h}) \geq h}\\
    &= \sum_{\tilde{h} \in \mathsf{support}(H)}\Pr[H = \tilde{h}] \cdot \mathbbm{1}_{-\log \max_x \Pr[X = x| E, H = \tilde{h}] \geq h}
\end{align*}
The same convention holds for average conditional min-entropy. For example,
\begin{align*}
    \Pr_H\left[H_\infty^{\mathsf{avg}}(X|C, E) \geq h\right] &= \sum_{\tilde{h} \in \mathsf{support}(H)}\Pr[H = \tilde{h}] \ifSubmission{\\&\quad\quad}\fi\cdot \mathbbm{1}_{-\log \sum_{c} \max_x \Pr[X = x, C = c | E, H = \tilde{h}] \geq h}
\end{align*}

Next, Hoeffding's inequality says that the sum of bounded independent random variables is tightly concentrated around its mean.
\begin{theorem}[Hoeffding's Inequality, \cite{Hoe63} Theorem 2]\label{thm:hoeffding-inequality}
    Let $t > 0$, let $(a_1, \ldots, a_n, b_1, \ldots, b_n)$ be constants, and let $(X_1, \dots, X_n)$ be independent random variables such that for all $i \in [n]$, $a_i \leq X_i \leq b_i$. Next, let $S = X_1 + \ldots + X_n$. Then,
    \[\Pr\left[S \geq \bbE[S] + t\right] \leq e^{-\frac{2 t^2}{\sum_{i \in [n]} (b_i - a_i)^2}}\]
\end{theorem}

\subsection{Quantum Random Oracle Model} \label{subsec:qrom}
\footnote{This section comes from \cite{CKMRST26} verbatim.} Bellare and Rogaway \cite{BR93} defined the random oracle as a black-box oracle providing access to a uniformly random function whose inputs and outputs are arbitrarily long. The oracle $H$ takes as input a binary string $x \in \bit^*$ of any finite length, and outputs a uniformly random binary string $H(x)$ of any desired length. Furthermore, quantum algorithms and adversaries can query the random oracle on a quantum superposition of inputs, and the oracle responds coherently. This is known as the quantum random oracle model (QROM) \cite{BDFLSZ11}.

For notational convenience, we will often deal with random oracles that have a finite domain and range or constructions that use two independent random oracles. These primitives can be constructed in the random oracle model, as it was defined above, using natural encoding schemes.

\subsection{Oracle-Aided Quantum Algorithms}\label{sec:prelim-oracle-aided-algorithms}
Here we define the notion of query weight for oracle-aided algorithms and provide the swapping lemma, which allows us to analyze such algorithms. This section is adapted from \cite{KRT26}.

Let $\calA^H$ denote a quantum algorithm $\calA$ with quantum query access to oracle $H$. $\calA$ is parametrized by $Q \geq 1$, which is the maximum number of quantum queries $\calA$ can make.

Without loss of generality, let $\calA^H$ operate as follows: 
\begin{itemize}
    \item $\calA$ starts with an initial pure state $\ket{\psi_0} = \ket{\psi^H_0}$ on registers $R = R_{Q} \times R_A$, a sequence of unitaries $(U_1, \dots, U_Q)$, and a measurement operator $M$.
    \item For each layer $j \in [Q]$:
    \begin{itemize}
        \item $\calA$ submits the query register $R_{Q}$ of $\ket{\psi_{j-1}^H}$ to the oracle $H$. Then $H$ acts as a phase oracle on $R_{Q}$ and returns $R_{Q}$ to $\calA$.
        \item $\calA$ applies $U_j$ to its state to obtain a new state $\ket{\psi_{j}^H}$.
    \end{itemize}
    \item Finally, $\calA$ applies measurement $M$ to $\ket{\psi_{Q}^H}$ and outputs the measurement outcome.
\end{itemize}

Next, we define the query weight $w_{x}^{j}(H)$ to be the probability that we obtain query $x$ if we measure $\calA^H$'s $j$-th query.

\begin{definition}[Query Weight]
For each possible classical query $x$ and each query $j \in [Q]$, let $w_{x}^{j}(H)$ be the \textbf{query weight} that $\calA^H$ gives to $x$ on the $j$-th query.
\[w_{x}^{j}(H) = \Tr\left[\left(\ketbra{x}_{R_{Q}} \otimes \mathbb{I}_{R_A}\right) \ketbra{\psi_{j-1}^H}\right]\]
Also, let $w_{x}^0(H) = 0$ for all $(x, H)$. Finally, let $w_{x}^{\leq j}(H)$ be the \textbf{cumulative query weight} that $\calA^H$ gives to $x$ on the first $j$ queries.
\[w_{x}^{\leq j}(H) = \sum_{j' = 0}^j w_{x}^{j'}(H)\]
\end{definition}

The following lemma gives some basic properties of the query weight.
\begin{lemma}[\cite{KRT26} Lemmas 3.2-3.5]\label{thm:query-weight-properties}
For any $\calA, H$ and any $j', j \in [Q]$ such that $j' < j$,
\begin{enumerate}
    \item For any $x$, $0 \leq w_{x}^{j}(H)$
    \item For any $x$, $w_{x}^{\leq j'}(H) \leq w_{x}^{\leq j}(H)$
    \item $\sum_{x} w^{j}_{x}(H) = 1$
    \item $\sum_{x} w^{\leq j}_{x}(H) = j$
\end{enumerate}
\end{lemma}

The swapping lemma (\cref{thm:swapping-lemma}) says that the adversary's states on two oracles $\RO$ and $\RO'$ will be close in Euclidean distance if the adversary gives small query weight to the positions where the oracles differ.
 \begin{lemma}[Swapping Lemma, \cite{KRT26} Lemma 3.6, Adapted from \allowbreak\cite{Vaz98} Lemma 2,
\cite{BBBV97} Theorem 3.3]\label{thm:swapping-lemma}
        Given two oracles $\RO, \RO'$, let $X$ be the subset of inputs on which $\RO$ and $\RO'$ differ. Then, for any $j \in \{0, \ldots, Q\}$,
        \[\left\|\ket{\psi^\RO_{j}} - \ket{\psi^{\RO'}_{j}}\right\|_2 \leq 2 \cdot \sqrt{j \cdot \sum_{x \in X} w_{x}^{\leq j}(\RO)}\]
    \end{lemma}

The following lemma says that the Euclidean distance between two states upper-bounds the trace distance.
    \begin{lemma}[\cite{KRT26} Lemma~3.7]\label{thm:trace-dist-euclidian-dist}
        For any quantum pure states $\ket{\psi}$ and $\ket{\phi}$, \[\mathsf{TraceDist}\left(\ket{\psi}, \ket{\phi}\right) \leq \left\|\ket{\psi} - \ket{\phi}\right\|_2\]
    \end{lemma}

\subsection{Proof of Min-Entropy}
Here we define a proof of min-entropy, which is a protocol to generate a random string $x$ and prove that it was generated according to a procedure with high min-entropy. \Cref{def:proof-of-min-entropy-YZ} is adapted from \cite[Definition~3.5]{YZ24} and \cite[Definition~6.1]{CKMRST26}. However, we use a stronger notion of soundness than was considered in prior work, requiring an extremely negligible soundness parameter.

\begin{definition}[Proof of Min-Entropy]\label{def:proof-of-min-entropy-YZ}
    A (keyless, non-interactive, publicly verifiable) \textbf{proof of min-entropy} relative to a random oracle  consists of the following algorithms $(\Prove, \Verify)$. \\
    
    \noindent$\Prove^{H}(1^\lambda,1^h) \to \pi:$ This is a QPT algorithm that takes as input the security parameter $1^{\lambda}$ and a min-entropy threshold $1^h$. Then $\Prove$ makes $\text{poly}(\lambda, h)$ quantum queries to the random oracle $H$, and outputs a classical proof $\pi$.\\
  
  \noindent$\Verify^{H}(1^{\lambda},1^h,\pi) \to x:$ This is a deterministic classical polynomial time algorithm that takes as input the security parameter $1^{\lambda}$, the min-entropy threshold $1^h$, and a proof $\pi$, then makes $\text{poly}(\lambda, h)$ queries to $H$, and outputs either a bitstring $x \in \bit^n$, whose length $n$ may depend on $h$, or $\bot$ indicating rejection.\\

  We require a proof of min-entropy to satisfy the following properties:\\
  
  \noindent\textbf{Correctness:} For any function $h = h(\secp)$, there is a negligible function $\negl$ such that for any $\secp \in \bbN$, we have
  \begin{equation*}
      \Pr_{H}\left[\Verify^{H}(1^{\lambda},1^h,\pi) = \bot : \pi \gets \Prove^{H}(1^\lambda,1^h) \right] \leq \negl(\secp).
  \end{equation*}
  
  \noindent\textbf{Strong Soundness:} For any polynomially bounded functions $h=h(\lambda)$, $Q = Q(\secp)$, and any inverse-polynomial function $\delta$, there exists a negligible function $\nu$ such that for any unbounded-time adversary $\mathcal{A}$ that makes at most $Q$ quantum queries to $H$, the following holds.
  
  Let $\mathcal{A}^{H}_{\top}(1^{\lambda}, 1^h)$ be the distribution of $\Verify^H(1^{\lambda}, 1^{h}, \mathcal{A}^{H}(1^{\lambda}, 1^h))$ conditioned on the output not being $\bot$. Then, for any $\secp \in \bbN$,
  \[
  \Pr_{H}\left[\begin{array}{cc}
       &  \Pr[\Verify^{H}(1^{\lambda}, 1^h,\mathcal{A}^{H}(1^{\lambda}, 1^h))\neq \bot]\geq \delta (\lambda)\\
       & \wedge H_{\infty}\Big( \mathcal{A}^{H}_{\top}(1^{\lambda}, 1^h)\Big)\leq h(\lambda)
  \end{array} \right]\leq \nu(\lambda)^{h(\secp)}
  \]
\end{definition}

In this paper, all proofs of min-entropy are assumed to be keyless, non-interactive, and publicly verifiable unless otherwise stated.

\begin{theorem}\label{thm:YZ-proof-of-min-entropy}
    There exists a proof of min-entropy relative to a random oracle satisfying \cref{def:proof-of-min-entropy-YZ}. In particular, the protocol satisfies strong soundness.
\end{theorem}
\begin{proof}
    We will show that \cite[Construction~6.4]{CKMRST26} satisfies the properties of \cref{def:proof-of-min-entropy-YZ}, including strong soundness. \cite{CKMRST26} proved a weaker soundness condition, but we can tighten their analysis to prove the strong soundness of their scheme. 
    
    First, to allow us to compare the strong soundness of \cref{def:proof-of-min-entropy-YZ} to the certifiable min-entropy property of \cite[Definition~6.1]{CKMRST26}, let us define the \textbf{soundness parameter} to be the following probability, as a function of $\secp$:
    \[\Pr_{H}\Big[\Pr[\Verify^{H}(1^{\lambda}, 1^h,\mathcal{A}^{H}(1^{\lambda}, 1^h))\neq \bot]\geq \delta (\lambda) \wedge H_{\infty}\Big( \mathcal{A}^{H}_{\top}(1^{\lambda}, 1^h)\Big)\leq h(\lambda)\Big]\]
    
    Second, constructions 6.2 and 6.4 of \cite{CKMRST26} satisfy the syntax and correctness properties of a proof of min-entropy (Theorems 6.3 and 6.5 of \cite{CKMRST26}). Construction 6.2 provides $h(\secp) = o(\secp^c)$ bits of min-entropy for $Q = 2^{o\left(\secp^c\right)}$ with soundness parameter
    \[\beta(\secp) = e^{-\Omega(\sqrt{\secp})}\]
    This follows from the proof of \cite[Theorem 7.1]{CKMRST26}. The proof of that theorem shows that for any $h(\secp) = o(\secp^c)$, $Q = 2^{o\left(\secp^c\right)}$, and any inverse-polynomial function $\delta$, the soundness parameter $\beta$ satisfies:
    \[\left[\beta(\secp) - L\cdot (1-p)^{4n/5}\right] \cdot \alpha(\secp) \leq 2\cdot (1-p)^{s}\]
    for any $Q$-query adversary.
    In this proof, the parameters satisfy $L = 2^{O(\secp^c \log \secp)}$, $\alpha(\secp) = \delta(\secp) \cdot 2^{-h(\secp)}$, $p = \Theta\left(n^{-1/2}\right)$, $s = \lceil{n/10\rceil} = \Theta(n)$, $n(\secp) = 2^{\lfloor\log \lambda \rfloor}-1 = \Theta(\secp)$, $c \in \left(0, \frac{1}{2}\right)$.
    This implies that for sufficiently large $\secp$,
    \begin{align*}
        \beta(\secp) &\leq L\cdot (1-p)^{4n/5} + \frac{2}{\alpha(\secp)} \cdot (1-p)^{s}\\
        &= L\cdot (1-p)^{4n/5} + 2 \cdot \delta(\secp)^{-1} \cdot 2^{h(\secp)} \cdot (1-p)^{s}\\
        &\leq 2^{O(\secp^c \log \secp)} \cdot e^{-\frac{p \cdot 4n}{5}} + 2 \cdot 2^{\secp^{c}} \cdot 2^{\secp^{c}} \cdot e^{-p \cdot s}\\
        &\leq 2^{O(\secp^c \log \secp)} \cdot e^{-\Theta(\sqrt{\secp})} + 2 \cdot 2^{\secp^{c}} \cdot 2^{\secp^{c}} \cdot e^{-\Omega(\sqrt{\secp})}\\
        &= e^{-\Omega(\sqrt{\secp})}
    \end{align*}

    Third, Construction 6.4 uses complexity leveraging by running Construction 6.2 with security parameter $\max\{\secp, h(\secp)^{2/c}\}$. The result is that for any $h = \poly(\secp)$, $Q = \poly(\secp)$, and inverse-polynomial $\delta$, Construction 6.4 achieves soundness parameter
    \begin{align*}
        \beta'(\secp) &= e^{-\Omega\left(\sqrt{\max\{\secp, h(\secp)^{2/c}\}}\right)} \leq e^{-\Omega\left(\sqrt{\sqrt{\secp} \cdot \sqrt{h(\secp)^{2/c}}}\right)} = e^{-\Omega\left(\secp^{1/4} \cdot h(\secp)^{1/(2c)}\right)}\\
        &\leq e^{-\Omega\left(\secp^{1/4} \cdot h(\secp)\right)} = \left(e^{-\Omega\left(\secp^{1/4}\right)}\right)^{h(\secp)}
    \end{align*}
    for any $Q$-query adversary.
    
    Fourth, let $\nu(\secp) = e^{-\secp^{1/5}}$ for sufficiently large $\secp$. Then the soundness parameter of Construction 6.4 satisfies $\beta'(\secp) \leq \nu(\secp)^{h(\secp)}$. Additionally, $\nu(\secp)$ is negligible in $\secp$.
\end{proof}

\begin{lemma}\label[lemma]{lem:Pr-X-and-C-is-small}
    For any polynomially bounded functions $h=h(\lambda)$, $Q = Q(\secp)$, and any inverse-polynomial function $\delta$, there exists a negligible function $\nu$ such that for any unbounded-time adversary $\mathcal{A}$ that makes at most $Q$ quantum queries to $H$, the following holds.

    Let
    \begin{align*}
        (\Pi, C) &\gets \calA^H(1^\secp, 1^{h(\secp)})\\
        X &= \Verify^H(1^\secp, 1^{h(\secp)}, \Pi)
    \end{align*}
    Then for any $\secp \in \bbN$,
    \[
  \Pr_{H}\Big[\Pr[X \neq \bot] \geq \delta (\lambda) \land \max_{x, c} \Pr[X = x, C = c | X \neq \bot] \geq 2^{-h(\lambda)}\Big] \leq \nu(\lambda)^{h(\secp)}
  \]
\end{lemma}
\begin{proof}
    First, let us define adversary $\calA'^H\left(1^\secp, 1^{h(\secp)}\right)$ to run $\calA^H\left(1^\secp, 1^{h(\secp)}\right) \to (\Pi, C)$ and then output $\Pi$. $\calA'$ is an adversary for the strong soundness of $\mathsf{MinEnt}$ (\cref{def:proof-of-min-entropy-YZ}) that makes at most $Q$ quantum queries to $H$. By the strong soundness property of $\mathsf{MinEnt}$,
\begin{align*}
    \nu(\lambda)^{h(\secp)} & \geq \Pr_{H}\left[\begin{array}{cc}
         &\Pr[\mathsf{MinEnt}.\Verify^{H}(1^{\lambda}, 1^{h},\mathcal{A}'^{H}(1^{\lambda}, 1^{h}))\neq \bot]\geq \delta (\lambda)  \\
         &\wedge H_{\infty}\Big(\mathcal{A}'^{H}_{\top}(1^{\lambda}, 1^{h})\Big)\leq h(\lambda) 
    \end{array}\right]\\
    &= \Pr_{H}\left[\begin{array}{cc}
         &\Pr[X \neq \bot]\geq \delta (\lambda)  \\
         &\wedge H_{\infty}\Big(X | X \neq \bot \Big)\leq h(\lambda) 
    \end{array}\right]
\end{align*}
In the expression above $H_{\infty}\Big(X | X \neq \bot \Big)$ is really $H_{\infty}\Big(X | X \neq \bot, H = \tilde{h} \Big)$ for the value of $\tilde{h}$ that $H$ takes.

Second, for any $\tilde{h} \in \mathsf{support}(H)$, if $\max_{x,c} \Pr[X = x, C = c | X \neq \bot, H = \tilde{h}] \geq 2^{-h}$, then
\begin{align*}
    2^{-h} &\leq \max_{x,c} \Pr[X = x, C = c | X \neq \bot, H = \tilde{h}]\\
    &\leq \max_x \Pr[X = x | X \neq \bot, H = \tilde{h}]\\
    h &\geq H_\infty(X | X \neq \bot, H = \tilde{h})
\end{align*}

Third, let us put everything together.
\begin{align*}
    \nu(\lambda)^{h(\secp)} &\geq \Pr_{H}\left[\begin{array}{cc}
         &\Pr[X \neq \bot] \geq \delta (\lambda)  \\
         &\land H_{\infty}\Big(X | X \neq \bot \Big) \leq h(\lambda) 
    \end{array}\right]\\
    &\geq \Pr_{H}\left[\begin{array}{cc}
         &\Pr[X \neq \bot] \geq \delta (\lambda)  \\
         &\land \max_{x,c} \Pr[X = x, C = c | X \neq \bot] \geq 2^{-h(\lambda)} 
    \end{array}\right]
\end{align*}
\end{proof}
\section{Proof of Maximal Conditional Min-Entropy}\label{sec:proof-of-maximal-conditional-min-entropy}

\subsection{Definition}

\begin{definition}[Proof of maximal conditional min-entropy]\label{def:adaptive}
A (keyless, non-interactive, publicly verifiable) proof of maximal conditional min-entropy, relative to a random oracle, consists of the following algorithms $(\Prove, \Verify)$:\\

\noindent
$\mathsf{Prove}^{H}(1^{\lambda},1^n,c)\longrightarrow \pi:$ This is a QPT algorithm that takes as input the security parameter $1^{\lambda}$, a label $c\in\bits^{\ell}$ and an output length $1^n$. Then $\Prove$ makes $\text{poly}(\lambda, n, \ell)$ quantum queries to the random oracle $H$, and outputs a classical proof $\pi$.\\

\noindent $
\mathsf{Verify}^{H}(1^{\lambda},1^n,c,\pi)
\longrightarrow x:
$
This is a \emph{deterministic} classical polynomial time algorithm that takes as input the security parameter $1^{\lambda}$, the output length $1^n$, label $c\in \{0,1\}^{\ell}$, and a proof $\pi$, and then makes $\text{poly}(\lambda, n, \ell)$ queries to $H$, and outputs either a bitstring $x \in \bit^n$, or $\bot$ indicating rejection.\\

\noindent We require a proof of maximal conditional min-entropy to satisfy the following properties:\\

\noindent\textbf{Correctness.}  For every polynomially bounded $n=n(\lambda)$ and $\ell=\ell(\lambda)$, there exists a negligible function $\negl$ such that, for every $c\in\bits^{\ell}$,
$$
\Prob_H\!
\left[
\mathsf{Verify}^{H}(1^{\lambda},1^n,c,\pi)=\bot
:
\pi\leftarrow\mathsf{Prove}^{H}(1^{\lambda},1^n,c)
\right]
\leq \negl(\lambda).
$$

\noindent\textbf{Soundness.}  Let $n=n(\lambda)=\omega(\log\lambda)$, let $\ell=\ell(\lambda)$ and $Q=Q(\lambda)$ be polynomially bounded, let $\gamma(\lambda)=\omega(\log\lambda)$, and let $\delta= \delta(\lambda)$ be inverse polynomial.  For every unbounded-time adversary $A$ making at most $Q$ quantum  queries to $H$, define
$$
(C,\Pi)\leftarrow A^{H}(1^{\lambda},1^{n(\lambda)}),
$$
$$
X:=\mathsf{Verify}^{H}(1^{\lambda},1^{n(\lambda)},C,\Pi),
\qquad
\Acc:=\{X\neq\bot\}.
$$
Then there exists a negligible function $\negl$ such that for all $\secp \in \bbN$,
\begin{equation}
\Prob_H\!
\left[
\Prob[\Acc]\geq\delta(\lambda)
\;\wedge\;
\Havg(X\mid C,\Acc)\leq n(\lambda)-\gamma(\lambda)
\right]
\leq \negl(\lambda).
\end{equation}

\end{definition}

\subsection{Construction}
\textit{Parameters}: When $\Prove$ and $\Verify$ receive the input  parameters $(\secp, n, \ell)$, they let $\lambda^{\star}$ be the amplified inner security parameter $ \lambda^{\star}:=\lambda+3n+\ell+10$,  
let $h = 4 n+ \ell+ \lambda^\star+2$, and let $m$ be the output length of $\mathsf{MinEnt}.\Verify$ for the given values of $(\lambda^{\star}, h)$.

\begin{construction}[Proof of maximal conditional min-entropy]\label{constr:maxent}
Let
$$
(\mathsf{MinEnt.Prove},\mathsf{MinEnt.Verify})
$$
be the unlabelled proof of min-entropy supplied by \cref{thm:YZ-proof-of-min-entropy} instantiated at parameters $(\lambda^{\star},h)$.
This construction uses two independent random oracles.  The inner random oracle is $H$.  The outer random oracle is
$G:\bits^{\ell}\times\bits^m\longrightarrow\bits^n.$
The two oracles can be implemented by one global random oracle using prefix-free domain separation, as in section \ref{subsec:qrom}.
Now we construct $(\Prove, \Verify)$ below. 
\begin{itemize}
\item $\mathsf{Prove}^{G,H}(1^{\lambda},1^n,c)$: Run
$
\pi\leftarrow
\mathsf{MinEnt.Prove}^{H}(1^{\lambda^{\star}},1^h)
$
and output $\pi$.

\item 
$\mathsf{Verify}^{G,H}(1^{\lambda},1^n,c,\pi):$
Compute $
w\leftarrow
\mathsf{MinEnt.Verify}^{H}(1^{\lambda^{\star}},1^h,\pi)
$. 
If $w=\bot$, output $\bot$.  Otherwise output
$
x:=G(c,w).
$
\end{itemize}
\end{construction}

\begin{theorem}[Syntax and Correctness]\label{thm:syntax}
The construction above satisfies the syntax and correctness requirements of Definition~\ref{def:adaptive}.
\end{theorem}

\begin{proof}
    The parameters $h$ and $\lambda^{\star}$ are polynomial in the external parameters.  By the syntax of the inner construction denoted by $(\mathsf{MinEnt.Prove},\mathsf{MinEnt.Verify})$, its prover and verifier run in time $\poly(\lambda^{\star},h,\ell)=\poly(\lambda,n,\ell)$ and make polynomially many queries to $H$.  The outer verifier makes one additional classical query to $G$.  Hence the syntax requirement holds.

Fix a label $c\in\bits^{\ell}$.  The outer verifier rejects if and only if the unlabelled inner verifier $\mathsf{MinEnt.Verify}$ rejects.  Therefore
$$
\begin{aligned}
&\Prob_{G,H}\!\left[
\mathsf{Verify}^{G,H}
(1^{\lambda},1^n,c,
\mathsf{Prove}^{G,H}(1^{\lambda},1^n,c))
=\bot
\right]
\\
&\qquad=
\Prob_H\!\left[
\mathsf{MinEnt.Verify}^{H}
(1^{\lambda^{\star}},1^h,
\mathsf{MinEnt.Prove}^{H}(1^{\lambda^{\star}},1^h))
=\bot
\right].
\end{aligned}
$$
The right-hand side is negligible by correctness of the inner construction.  The random choice of $G$ is irrelevant to rejection, so the equality is exact.  This proves correctness.
\end{proof}

\ifSubmission
    In \Cref{sec:proof-of-soundness} we prove that Construction \ref{constr:maxent} satisfies the soundness property of Definition \ref{def:adaptive}.
\fi
\ifSubmission\else
    \ifSubmission
    \section{Proof of Soundness}\label[section]{sec:proof-of-soundness}
\else
    \subsection{Proof of Soundness}\label[section]{sec:proof-of-soundness}
\fi
Here we prove that Construction \ref{constr:maxent} satisfies the soundness property of Definition \ref{def:adaptive} (\Cref{thm:adaptive-maximal-min-entropy}).

Fix polynomially bounded functions $n=n(\lambda)$, $Q=Q(\lambda)$, and $\ell=\ell(\lambda)$ with $n(\lambda)=\omega(\log\lambda)$.  Fix $\gamma(\lambda)=\omega(\log\lambda)$ and a function $\delta= 1/\poly(\lambda)$ .  Let $A^{G,H}$ be an arbitrary, possibly unbounded-time adversary making at most $Q$ total quantum queries to the pair $(G,H)$.  Define
$$
(C,\Pi)
\leftarrow
A^{G,H}(1^{\lambda},1^{n}),
$$
$$
W
:=
\mathsf{MinEnt.Verify}^{H}(1^{\lambda^{\star}},1^h,\Pi),
$$
$$
X
:=
\begin{cases}
G(C,W),&W\neq\bot,\\
\bot,&W=\bot.
\end{cases}
$$
Thus $X\neq\bot$ if and only if $W\neq\bot$.

\noindent The rest of the section proves that
$$
\Prob_{G,H}\!\left[
\Prob[X\neq\bot\mid G,H]\geq\delta
\;\wedge\;
\Havg(X\mid C,X \neq \bot, G,H)\leq n-\gamma
\right]\leq \negl(\lambda)
$$
 \\
We repeatedly use Hoeffding's inequality in the following standard form.  If $Z_1,\ldots,Z_N$ are independent random variables satisfying $a_i\leq Z_i\leq b_i$ almost surely and $S=\sum_i Z_i$, then for every $t>0$,
$$
\Prob[S\geq\Exp[S]+t]
\leq
\exp\!\left(-\frac{2t^2}{\sum_{i=1}^{N}(b_i-a_i)^2}\right).
$$

\subsubsection{Query normal form and query weights}
\noindent Let $R$ be a classical variable sampled independently of a fresh random function
$
G:\bits^m\longrightarrow\bits^n.
$
The value $R=r$ may specify arbitrary auxiliary oracles, arbitrary classical side information, and an arbitrary initial mixed state.  Conditioned on $R=r$, an algorithm $A^{G,r}$ makes at most $q$ quantum queries to $G$ and outputs
$$
W\in\bits^m\cup\{\bot\}.
$$
Define
$$
Y:=
\begin{cases}
G(W),&W\neq\bot,\\
\bot,&W=\bot.
\end{cases}
$$
The variable $Y$ is used only in this subsection; the final construction uses the notation $X$ defined earlier.\\

\noindent Without loss of generality, conditioned on $R=r$, $A^{G,r}$  has the state $|\psi^{g,r}_0\rangle$ and a sequence of oracle-independent unitaries.  Immediately before query layer $j\in[q]$, its state is $|\psi^{g,r}_{j-1}\rangle$. 

\begin{lemma}[Swapping lemma, \cite{KRT26} Lemma 3.6, Adapted from \allowbreak\cite{Vaz98} Lemma 2,
\cite{BBBV97} Theorem 3.3]\label{lem:swap}
For any functions $g,g':\bits^m\rightarrow\bits^n$, any fixed $r$, and every $j\leq q$,
$$
\norm{|\psi^{g,r}_j\rangle-|\psi^{g',r}_j\rangle}
\leq
2\sqrt{
j\sum_{w:g(w)\neq g'(w)}
\wt^{\leq j}_w(g,r)
}.
$$
\end{lemma}

\subsubsection{The \texorpdfstring{target-$x$}{target-x} completion procedure}

Fix a target output value
$$
x\in\bits^n.
$$
The following procedure samples two correlated functions $G,G'$ with uniform marginals.  They agree on every input that becomes noticeably queried and are resampled independently only on inputs that remain light.

\begin{definition}[Target-$x$ completion procedure]\label{def:completion}
Perform the following experiment.

\begin{enumerate}
\item Sample $R$ according to its prescribed distribution.

\item Independently for every $w\in\bits^m$, sample
$$
E_0(w)\xleftarrow{\$}\bits^n\setminus\{x\}.
$$

\item Set $D_0(w)=0$ for every $w$, set $\Delta_0=\varnothing$, and initialize the algorithm.

\item For every query layer $j\in[q]$, do the following.

\begin{enumerate}
\item Let $|\phi_{j-1}\rangle$ be the current hybrid state immediately before the $j$-th query.  Define
$$
\wt^{j}_w
:=
\Tr\!\left[
(\ketbra{w}{w}\otimes I)
\ketbra{\phi_{j-1}}{\phi_{j-1}}
\right].
$$

\item Define the new heavy set
$$
\Delta_j
:=
\left\{
 w\notin\bigcup_{i<j}\Delta_i
:
\wt^{ j}_w\geq \frac{2^{-3n}}{q^2}
\right\}.
$$

\item For every $w\in\Delta_j$, independently sample a Bernoulli variable $D_j(w)$ satisfying
$$
\Prob[D_j(w)=1]=2^{-n}.
$$
For every previously exposed coordinate retain its old value, and for every coordinate not yet exposed retain the placeholder value zero.

\item Define the oracle used at layer $j$ by
$$
E_j(w)
:=
\begin{cases}
x,&D_j(w)=1,\\
E_0(w),&D_j(w)=0.
\end{cases}
$$
Apply $E_j$ at the $j$-th query and then apply the algorithm's next oracle-independent unitary, obtaining $|\phi_j\rangle$.
\end{enumerate}

\item Let
$$
F:=E_q.
$$
Define the light set
$$
\Delta_{q+1}
:=
\bits^m\setminus\bigcup_{j=0}^{q}\Delta_j.
$$

\item For every $w\in\Delta_{q+1}$, sample two independent Bernoulli variables $D(w)$ and $D'(w)$, each of mean $2^{-n}$.  For $w\notin\Delta_{q+1}$, set
$$
D(w)=D'(w)=D_q(w).
$$

\item Define the completed functions
$$
G(w)
:=
\begin{cases}
x,&D(w)=1,\\E_0(w),&D(w)=0,
\end{cases}
\qquad
G'(w)
:=
\begin{cases}
x,&D'(w)=1,\\E_0(w),&D'(w)=0.
\end{cases}
$$
\end{enumerate}
\end{definition}

 On the heavy coordinates the two bits are the same; on the light coordinates they are independent.

\subsubsection{Uniformity of the completed functions}

\begin{lemma}[Uniform marginals]\label{lem:uniform}
Conditioned on every value $R=r$, each of $G$ and $G'$ in Definition~\ref{def:completion} is a uniformly random function  $\bits^m \rightarrow \bits^n$.
\end{lemma}

\begin{proof}
Let 
$$
p:=2^{-n},
\qquad
N:=2^m.
$$
We first prove that, conditioned on $(R,E_0)=(r,e_0)$, the final bit table $D$ has the product Bernoulli-$p$ distribution.

Fix an arbitrary deterministic ordering of the coordinates inside every exposed batch $\Delta_j$.  This induces a sequence
$$
w_1,w_2,\ldots,w_N
$$
containing every point of $\bits^m$ exactly once.  $w_k$ may depend on $(r,e_0)$ and on bits exposed at earlier positions, because the query weights and the sets $\Delta_j$ are adaptive.  Crucially, before $w_k$ is exposed, its Bernoulli bit has not been sampled and is independent of the entire previous history.

Fix a candidate table $d:\bits^m\rightarrow\bits$.  For $k\geq 1$, abbreviate the previously exposed history by
$$
\mathcal H_{k-1}(d)
:=
\{R=r,E_0=e_0,D(w_1)=d(w_1),\ldots,D(w_{k-1})=d(w_{k-1})\}.
$$
Conditional on $\mathcal H_{k-1}(d)$,
\begin{equation}
\Prob[D(w_k)=d(w_k)\mid\mathcal H_{k-1}(d)]= \begin{cases}
    p  \text{    when } d(w_k)=1\\
    1-p \text{ when } d(w_k)=0
\end{cases} 
\end{equation}
 Thus,
$$
\begin{aligned}
\Prob[D=d\mid R=r,E_0=e_0]
&=
\prod_{k=1}^{N}
\Prob[D(w_k)=d(w_k)\mid\mathcal H_{k-1}(d)]
\\
&=
p^{|\{w:d(w)=1\}|}(1-p)^{|\{w:d(w)=0\}|}.
\end{aligned}
$$
The right-hand side is the product Bernoulli probability and does not depend on $e_0$.  Therefore $D$ is independent of $E_0$ conditional on $R=r$, and its coordinates are mutually independent Bernoulli-$p$ variables.

\noindent The same argument applies to $D'$.  On a heavy coordinate, $D'$ receives the same fresh bit that was exposed during the execution.  On a light coordinate, $D'$ receives a fresh independent bit at the final completion step.  Every coordinate is still exposed exactly once for the marginal $D'$, and the chain-rule calculation is identical.  Thus $D'$ also has the product Bernoulli-$p$ distribution and is independent of $E_0$, conditional on $R=r$.

\noindent Now fix a coordinate $w$.  For the target output $x$,
$$
\Prob[G(w)=x\mid R=r]
=
\Prob[D(w)=1\mid R=r]
=p
=2^{-n}.
$$
For any $y\neq x$,
$$
\begin{aligned}
\Prob[G(w)=y\mid R=r]
&=
\Prob[D(w)=0\mid R=r]
\Prob[E_0(w)=y\mid D(w)=0,R=r]
\\
&=(1-2^{-n})\frac{1}{2^n-1}
\\
&=\frac{2^n-1}{2^n}\frac{1}{2^n-1}
\\
&=2^{-n}.
\end{aligned}
$$
Because the pairs $(D(w),E_0(w))$ are independent across coordinates, the outputs $G(w)$ are independent across coordinates.  Hence, for every fixed function $g:\bits^m\rightarrow\bits^n$,
$$
\Prob[G=g\mid R=r]
=
\prod_{w\in\bits^m}2^{-n}
=
2^{-n2^m},
$$
which is exactly the uniform distribution over all such functions.  Replacing $D$ by $D'$ proves the same statement for $G'$.
\end{proof}

\subsubsection{Distance between the two completed executions}

The hybrid state $|\phi_q\rangle$ was produced by a time-dependent oracle sequence $(E_1,\ldots,E_q)$.  To compare it with the static-oracle executions under $G$ and $G'$, we refer to the hybrid as a single oracle on the enlarged domain $[q]\times\bits^m$: on query layer $j$, the algorithm queries only points of the form $(j,w)$ and receives $E_j(w)$.  The query weight of $(j,w)$ is exactly $\wt_{j,w}$.

\begin{lemma}[Deterministic distance bound]\label{lem:det-distance}

$$
\norm{|\psi^{G',R}_q\rangle-|\psi^{G,R}_q\rangle}^{2}
\leq
8q
\sum_{j=2}^{q+1}
\sum_{w\in\Delta_j}
\bigl(D(w)+D'(w)\bigr)
\left(
\sum_{j'=1}^{j-1}\wt_{j',w}
\right).
$$
\end{lemma}

\begin{proof}
We first compare the hybrid state $|\phi_q\rangle$ with the  execution under $G$.  Lemma~\ref{thm:swapping-lemma} implies that
$$
\norm{|\phi_q\rangle-|\psi^{G,R}_q\rangle}^{2}
\leq
4q
\sum_{j'=1}^{q}
\sum_{w:E_{j'}(w)\neq G(w)}
\wt_{j',w}.
$$
Fix a coordinate $w\in\Delta_j$.  If $j\leq q$, then the Bernoulli bit at $w$ is first sampled immediately before query layer $j$.  For every earlier layer $j'<j$, the hybrid oracle still has value $E_0(w)$ at $w$.  For every layer $j'\geq j$, the hybrid oracle has already incorporated the sampled bit and therefore equals $G(w)$ at $w$.  Thus
$$
E_{j'}(w)\neq G(w)
$$
can occur only when $j'<j$, and in that case it occurs exactly when $D(w)=1$.  If $w\in\Delta_{q+1}$, its bit is sampled only after the final query, so the same statement holds for every $j'\leq q$.  Consequently,
$$
\sum_{j'=1}^{q}
\sum_{w:E_{j'}(w)\neq G(w)}
\wt_{j',w}
=
\sum_{j=2}^{q+1}
\sum_{w\in\Delta_j}
D(w)
\left(
\sum_{j'=1}^{j-1}\wt_{j',w}
\right).
$$
Therefore
$$
\norm{|\phi_q\rangle-|\psi^{G,R}_q\rangle}^{2}
\leq
4q
\sum_{j=2}^{q+1}
\sum_{w\in\Delta_j}
D(w)
\left(
\sum_{j'=1}^{j-1}\wt_{j',w}
\right).
$$
An identical argument with $G'$ in place of $G$ gives
$$
\norm{|\phi_q\rangle-|\psi^{G',R}_q\rangle}^{2}
\leq
4q
\sum_{j=2}^{q+1}
\sum_{w\in\Delta_j}
D'(w)
\left(
\sum_{j'=1}^{j-1}\wt_{j',w}
\right).
$$
By the triangle inequality,
$$
\norm{|\psi^{G',R}_q\rangle-|\psi^{G,R}_q\rangle}
\leq
\norm{|\psi^{G',R}_q\rangle-|\phi_q\rangle}
+
\norm{|\phi_q\rangle-|\psi^{G,R}_q\rangle}.
$$
Thus, 
\begin{align*}
&\norm{|\psi^{G',R}_q\rangle-|\psi^{G,R}_q\rangle}^{2}\\&\leq 8q \sum_{j=2}^{q+1}
\sum_{w\in\Delta_j}
D(w)
\left(
\sum_{j'=1}^{j-1}\wt_{j',w}
\right)  +8q \sum_{j=2}^{q+1}
\sum_{w\in\Delta_j}
D'(w)
\left(
\sum_{j'=1}^{j-1}\wt_{j',w}
\right)  \\
&= 8q \sum_{j=2}^{q+1}
\sum_{w\in\Delta_j}
(D(w)+D'(w))
\left(
\sum_{j'=1}^{j-1}\wt_{j',w}
\right)
\end{align*}
\end{proof}

\begin{lemma}[Lightness before exposure]\label{lem:prelight}
For every $j\in\{2,\ldots,q+1\}$, every $j'<j$, and every $w\in\Delta_j$,
$$
\wt_{j',w}<\frac{2^{-3n}}{q^2}.
$$
\end{lemma}

\begin{proof}
Because $w\in\Delta_j$, it was not selected in any earlier set $\Delta_1,\ldots,\Delta_{j-1}$.  In particular, at layer $j'$ it was neither previously selected nor newly selected.  By the definition of $\Delta_{j'}$, a previously unselected coordinate with weight at least $2^{-3n}/q^2$ would have been placed in $\Delta_{j'}$.  Therefore its weight at layer $j'$ must be strictly smaller than the threshold.
\end{proof}

\begin{lemma}[One-layer Hoeffding bound]\label{lem:layer-hoeffding}
Fix $j\in\{2,\ldots,q+1\}$.  Condition on the complete history immediately before the Bernoulli variables associated with $\Delta_j$ are sampled.  Then
$$
\Prob\!\left[
\sum_{w\in\Delta_j}
\bigl(D(w)+D'(w)\bigr)
\left(
\sum_{j'=1}^{j-1}\wt_{j',w}
\right)
\geq
(q+1)2^{-n+1}
\right]
\leq
\exp(-2^{n+1}).
$$
\end{lemma}

\begin{proof}
After conditioning on the indicated history, the set $\Delta_j$ and all weights $\wt^{ j'}_w$ for $j'<j$ are fixed constants.  Let
\begin{equation}
a_w:=\sum_{j'=1}^{j-1}\wt_{j',w}
\qquad
(w\in\Delta_j),
\end{equation}
\begin{equation}
Z_w:=\bigl(D(w)+D'(w)\bigr)a_w,
\qquad
S_j:=\sum_{w\in\Delta_j}Z_w.
\end{equation}
For distinct $w\in\Delta_j$, the variables $Z_w$ are independent.  When $j\leq q$, the two completed bits are equal on $\Delta_j$, so $D(w)+D'(w)$ is either $0$ or $2$.  When $j=q+1$, the two bits are independent, but their sum is still in $\{0,1,2\}$.  In both cases,
$$
0\leq Z_w\leq 2a_w
\qquad
$$
and
$$
\Exp[Z_w]=2\cdot 2^{-n}a_w=2^{-n+1}a_w.
$$
Consequently,

\begin{align*}
\Exp[S_j]
&=2^{-n+1}\sum_{w\in\Delta_j}a_w
\\
&=2^{-n+1}
\sum_{j'=1}^{j-1}
\sum_{w\in\Delta_j}\wt_{j',w}
\\
&\leq 2^{-n+1}(j-1)
\\
&\leq q2^{-n+1}.
\end{align*}

The first inequality uses the fact that the total query weight at every layer is one.

We next bound the sum of squared ranges.  By Lemma~\ref{lem:prelight},
$$
\max_{w\in\Delta_j}a_w
<
(j-1)\frac{2^{-3n}}{q^2}
\leq
\frac{2^{-3n}}{q}.
$$
Also,
$$
\sum_{w\in\Delta_j}a_w
\leq j-1\leq q.
$$
Therefore
$$
\begin{aligned}
\sum_{w\in\Delta_j}(2a_w-0)^2
&=4\sum_{w\in\Delta_j}a_w^2
\\
&\leq
4\left(\max_{w\in\Delta_j}a_w\right)
\sum_{w\in\Delta_j}a_w
\\
&\leq
4\cdot\frac{2^{-3n}}{q}\cdot q
\\
&=2^{-3n+2}.
\end{aligned}
$$
If
$$
S_j\geq(q+1)2^{-n+1},
$$
then
$$
S_j\geq \Exp[S_j]+2^{-n+1}.
$$
Hoeffding's inequality gives
$$
\begin{aligned}
\Prob[S_j\geq(q+1)2^{-n+1}]
&\leq
\exp\!\left(
-\frac{2(2^{-n+1})^2}{2^{-3n+2}}
\right)
\\
&=
\exp\!\left(-2^{n+1}\right).
\end{aligned}
$$

\end{proof}

\begin{lemma}[Distance tail]\label{lem:distance-tail}
For the target-$x$ completion procedure,
$$
\Prob\!\left[
\norm{|\psi^{G',R}_q\rangle-|\psi^{G,R}_q\rangle}^{2}
\geq
q^2(q+1)2^{-n+4}
\right]
\leq
q\exp(-2^{n+1}).
$$
\end{lemma}

\begin{proof}
For $j\in\{2,\ldots,q+1\}$, let $\mathsf E_j$ be the event that 
\[
\sum_{w\in\Delta_j}
\bigl(D(w)+D'(w)\bigr)
\left(
\sum_{j'=1}^{j-1}\wt_{j',w}
\right)
\geq
(q+1)2^{-n+1}
\]
Lemma~\ref{lem:layer-hoeffding} bounds the conditional probability of $\mathsf E_j$ by $\exp(-2^{n+1})$ for every possible prior history.  Averaging over that history preserves the same bound, thus
$$
\Prob[\mathsf E_j]\leq\exp(-2^{n+1}).
$$
A union bound over the $q$ indices $j=2,\ldots,q+1$ gives
$$
\Prob\!\left[\bigcup_{j=2}^{q+1}\mathsf E_j\right]
\leq q\exp(-2^{n+1}).
$$
Therefore, with probability at least $1-q\exp(-2^{n+1})$, Lemma \ref{lem:det-distance} implies that
$$
\begin{aligned}
\norm{|\psi^{G',R}_q\rangle-|\psi^{G,R}_q\rangle}^{2}
&<
8q\sum_{j=2}^{q+1}(q+1)2^{-n+1}
\\
&=
8q\cdot q\cdot(q+1)2^{-n+1}
\\
&=
q^2(q+1)2^{-n+4}.
\end{aligned}
$$

\end{proof}

\subsubsection{Projection onto one output value}

For a function $g:\bits^m\rightarrow\bits^n$, define the projector onto the preimage of the fixed target $x$ by
$$
\Pi_{g,x}
:=
\sum_{w\in\bits^m}
\mathbbm{1}_{g(w)=x} \cdot \ketbra{w}{w}\otimes \mathbbm{I}.
$$
The rejection outcome $\bot$ is represented in a subspace orthogonal to this projector.  Measuring the final output register of $|\psi^{g,r}_q\rangle$ gives $w$ with probability $\left\|\left(|w\rangle\langle w| \otimes \mathbbm{I}\right)\mid\psi^{g,r}_q\rangle\right\|^2$.  Hence the squared norm of the projection is
$$
\begin{aligned}
\norm{\Pi_{g,x}|\psi^{g,r}_q\rangle}^{2}
&=\sum_{w\in\bits^m}\mathbbm{1}_{g(w)=x} \cdot \left\|\left(|w\rangle\langle w| \otimes \mathbbm{I}\right)\mid\psi^{g,r}_q\rangle\right\|^2
\\
&=\Prob[Y=x\mid G=g,R=r].
\end{aligned}
$$

\begin{lemma}[Size of the heavy set]\label{lem:heavy-size}
For every outcome of the completion procedure,
$$
\left|\bigcup_{j=1}^{q}\Delta_j\right|
\leq q^3 2^{3n}.
$$
\end{lemma}

\begin{proof}
Fix a query layer $j$.  Every $w\in\Delta_j$ has query weight at least $2^{-3n}/q^2$ at that layer.  Hence
$$
|\Delta_j|\frac{2^{-3n}}{q^2}
\leq
\sum_{w\in\Delta_j}\wt_{j,w}
\leq 1,
$$
so
$$
|\Delta_j|\leq q^2 2^{3n}.
$$
The sets $\Delta_1,\ldots,\Delta_q$ are disjoint.  Summing  over  $q$ layers completes the proof.
\end{proof}

\begin{lemma}[Heavy-light projection bound]\label{lem:heavy-light}
Fix values $(\delta,g',r)$ in the support of $(\Delta_{q+1},G',R)$ and suppose
$$
\max_{w\in\bits^m}
\Prob[W=w\mid G'=g',R=r]
\leq 2^{-4n}.
$$
Then, conditioned on $(\Delta_{q+1},G',R)=(\delta,g',r)$, except with probability at most $\exp(-2^{2n+1})$,
$$
\norm{\Pi_{G,x}|\psi^{g',r}_q\rangle}^{2}
<
q^3 2^{-n}+2^{-n+1}.
$$
\end{lemma}

\begin{proof}
Let
$$
\mathsf H:=\bits^m\setminus\delta
=
\bigcup_{j=1}^{q}\Delta_j
$$
be the heavy set.  By Lemma~\ref{lem:heavy-size},
$$
|\mathsf H|\leq q^3 2^{3n}.
$$
Let
$$
b_w:=\left\|\left(|w\rangle\langle w| \otimes \mathbbm{I}\right)\mid\psi^{g',r}_q\rangle\right\|^2
=
\Prob[W=w\mid G'=g',R=r]
$$
for accepted $w\in\bits^m$.  The total mass on the heavy set is bounded by
$$
\sum_{w\in\mathsf H}b_w
\leq
|\mathsf H|\max_w b_w
\leq
q^3 2^{3n}\cdot 2^{-4n}
=
q^3 2^{-n}.
$$
This is an upper bound even if every heavy point happens to satisfy $G(w)=x$.

It remains to control the light set $\delta=\Delta_{q+1}$.  On every light coordinate, the bit $D(w)$ used by $G$ was sampled independently of the bit $D'(w)$ used by $G'$, independently of $E_0$, and independently across coordinates.  The set $\Delta_{q+1}$ is determined before these light bits are sampled.    Thus, even after conditioning on $(\Delta_{q+1},G',R)=(\delta,g',r)$, the family
$$
(D(w))_{w\in\delta}
$$
consists of independent Bernoulli variables of mean $2^{-n}$.

Define
$$
Z:=\sum_{w\in\delta}D(w)b_w.
$$
This is exactly the probability mass of the fixed state $|\psi^{g',r}_q\rangle$ on light coordinates for which $G(w)=x$.  Its conditional expectation is
$$
\begin{aligned}
\Exp[Z\mid\delta,g',r]
&=
\sum_{w\in\delta}b_w\Exp[D(w)\mid\delta,g',r]
\\
&=
2^{-n}\sum_{w\in\delta}b_w
\\
&\leq 2^{-n}.
\end{aligned}
$$
For each $w\in\delta$, the summand $D(w)b_w$ lies in the interval $[0,b_w]$.  Moreover,
$$
\begin{aligned}
\sum_{w\in\delta}b_w^2
&\leq
\left(\max_{w\in\delta}b_w\right)
\sum_{w\in\delta}b_w
\\
&\leq
2^{-4n}\cdot 1
\\
&=2^{-4n}.
\end{aligned}
$$
If $Z\geq 2^{-n+1}$, then
$$
Z\geq \Exp[Z\mid\delta,g',r]+2^{-n}.
$$
Hoeffding's inequality therefore gives
$$
\begin{aligned}
\Prob[Z\geq2^{-n+1}\mid\delta,g',r]
&\leq
\exp\!\left(
-\frac{2(2^{-n})^2}{\sum_{w\in\delta}b_w^2}
\right)
\\
&\leq
\exp\!\left(
-\frac{2\cdot2^{-2n}}{2^{-4n}}
\right)
\\
&=
\exp(-2^{2n+1}).
\end{aligned}
$$
This completes the proof.
\end{proof}

\subsubsection{A tail bound for one fixed output}

Define
$$
K_q:=2^6(q+1)^3.
$$
and 
define 
$$
\varepsilon_{\mathrm{flat}}
:=
\Prob_{G,R}\!\left[
\max_{w\in\bits^m}
\Prob[W=w\mid G,R]
>
2^{-4n}
\right].
$$

\begin{lemma}[One fixed output]\label{lem:fixed-output}
For every fixed $x\in\bits^n$,
$$
\Prob_{G,R}\!\left[
\Prob[Y=x\mid G,R]
\geq K_q2^{-n}
\right]
\leq
3q\exp(-2^{n+1})+\varepsilon_{\mathrm{flat}}.
$$
\end{lemma}

\begin{proof}
Run the coupling of Definition~\ref{def:completion} for the fixed target $x$.  By Lemma~\ref{lem:uniform}, both $(G,R)$ and $(G',R)$ have the correct marginal distribution.  Let $\mathsf E_{\mathrm{dist}}$ be the distance-tail event from Lemma~\ref{lem:distance-tail}.  Then
$$
\Prob[\mathsf E_{\mathrm{dist}}]
\leq q\exp(-2^{n+1}).
$$
Let $\mathsf E_{\mathrm{flat}}$ be the event that, 
\[
\max_{w\in\bits^m}
\Prob[W=w\mid G',R]
>
2^{-4n}
\]
 Since $(G',R)$ has the same marginal distribution as $(G,R)$,
$$
\Prob[\mathsf E_{\mathrm{flat}}]
=\varepsilon_{\mathrm{flat}}.
$$
On the complement of $\mathsf E_{\mathrm{flat}}$, Lemma~\ref{lem:heavy-light} applies.  Let $\mathsf E_{\mathrm{proj}}$ be its exceptional light-completion event.  Averaging its conditional bound over $(\Delta_{q+1},G',R)$ gives
$$
\Prob[\mathsf E_{\mathrm{proj}}\wedge\neg\mathsf E_{\mathrm{flat}}]
\leq \exp(-2^{2n+1}).
$$
Assume none of these three bad events occurs.  Then, 
$$
\begin{aligned}
\sqrt{\Prob[Y=x\mid G,R]}
&=
\norm{\Pi_{G,x}|\psi^{G,R}_q\rangle}
\\
&\leq
\norm{\Pi_{G,x}(|\psi^{G,R}_q\rangle-|\psi^{G',R}_q\rangle)}
+
\norm{\Pi_{G,x}|\psi^{G',R}_q\rangle}
\\
&\leq
\norm{|\psi^{G,R}_q\rangle-|\psi^{G',R}_q\rangle}
+
\norm{\Pi_{G,x}|\psi^{G',R}_q\rangle}.
\end{aligned}
$$
Set
$$
a:=\norm{|\psi^{G,R}_q\rangle-|\psi^{G',R}_q\rangle},
\qquad
b:=\norm{\Pi_{G,x}|\psi^{G',R}_q\rangle}.
$$
Using $(a+b)^2\leq2a^2+2b^2$, Lemma~\ref{lem:distance-tail}, and Lemma~\ref{lem:heavy-light},
$$
\begin{aligned}
\Prob[Y=x\mid G,R]
&\leq 2a^2+2b^2
\\
&<
2q^2(q+1)2^{-n+4}
+2\bigl(q^3 2^{-n}+2^{-n+1}\bigr)
\\
&=
\bigl(32q^2(q+1)+2q^3+4\bigr)2^{-n}
\\
&\leq
2^6(q+1)^3 2^{-n}
\\
&=K_q2^{-n}.
\end{aligned}
$$
The penultimate inequality holds for every $q\geq1$.

Therefore we can conclude that
$$
\begin{aligned}
\Prob\!\left[
\Prob[Y=x\mid G,R]\geq K_q2^{-n}
\right]
&\leq
q\exp(-2^{n+1})
+\varepsilon_{\mathrm{flat}}
+\exp(-2^{2n+1})
\\
&\leq
3q\exp(-2^{n+1})
+\varepsilon_{\mathrm{flat}}.
\end{aligned}
$$
This proves the lemma.
\end{proof}

\subsubsection{Hashing after the label is fixed}
\label{sec:lab}
The preceding subsection considers only one source $W$ and one fresh random function.  We now allow a first stage to output a label $C$ and a residual state before the fresh function is sampled.  The main issue is that joint point-flatness of $(C,W)$ does not imply point-flatness of $W$ conditioned on an arbitrarily rare label.  We handle rare labels by their total probability mass and apply the fixed-output lemma (Lemma \ref{lem:fixed-output}) only to non-rare labels.

Consider the following two-stage experiment.
\begin{enumerate}
\item Sample an auxiliary $H$ from an arbitrary distribution, and independently sample auxiliary classical side information $R$.  The variable $R$ may specify complete descriptions of other random oracles and arbitrary initial states.  It is sampled before the fresh slice below.  Both $R$ and $H$ are independent of the fresh function $\Theta$.

\item A first-stage algorithm, given $R$ and query access to $H$, outputs a classical label
$$
C\in\bits^{\ell}
$$
and a residual quantum state.

\item Independently of $(R,H)$, sample
$$
\Theta:\bits^m\longrightarrow\bits^n
$$
uniformly at random.

\item The second-stage algorithm receives quantum oracle access to $\Theta$, retains access to $H$ and the resources indexed by $R$, makes at most $q$ queries to $\Theta$, and outputs a proof $\Pi$.

\item Compute
$$
W:=\mathsf{MinEnt.Verify}^{H}(1^{\lambda^{\star}},1^h,\Pi),
$$
$$
Y:=
\begin{cases}
\Theta(W),&W\neq\bot,\\
\bot,&W=\bot.
\end{cases}
$$
\end{enumerate}

\noindent For fixed $(R,H)=(r,\widetilde h)$, define 
$$
r_c(r,\widetilde h)
:=
\Prob[C=c\mid R=r,H=\widetilde h].
$$
$C$ is produced before the independent function $\Theta$ is sampled, so for all $\theta$,
$$
\Prob[C=c\mid R=r,\Theta=\theta,H=\widetilde h]
=r_c(r,\widetilde h).
$$
For fixed $(r,\theta,\widetilde h)$, all probabilities below are over the randomness  of the two-stage algorithm.

\begin{lemma}\label{lem:fresh-slice}
For any $\beta,\varepsilon\in[0,1]$, suppose that
$$
\Prob_{R,\Theta,H}\!\left[
\max_{\substack{c\in\bits^{\ell}\\w\in\bits^m:w\neq\bot}}
\Prob[C=c,W=w\mid R,\Theta,H]
>
\beta
\right]
\leq\varepsilon,
$$
Then

\begin{equation}\label{eq1}
\begin{split}
&\Exp_{R,\Theta,H}\!\left[
\sum_{c\in\bits^{\ell}}
\max_{x\in\bits^n}
\Prob[C=c,Y=x\mid R,\Theta,H]
\right]
\\
&\quad\quad\quad\quad\leq
K_q2^{-n}
+2^{\ell+4n}\beta
+2^n\left[3q\exp(-2^{n+1})+\varepsilon\right].
\end{split}
\end{equation}

\end{lemma}

\begin{proof}
Set
$$
\tau:=2^{4n}\beta.
$$
For fixed $(r,\widetilde h)$, call a label $c$ \emph{small} when
$$
r_c(r,\widetilde h)<\tau.
$$
There are at most $2^{\ell}$ labels, so
\begin{equation}\label{eq:small_label}
\sum_{c:r_c(r,\widetilde h)<\tau}
r_c(r,\widetilde h)
\leq
2^{\ell}\tau
=
2^{\ell+4n}\beta.
\end{equation}
For every fixed $(r,\theta,\widetilde h)$ and every label $c$,
$$
\max_x\Prob[C=c,Y=x\mid r,\theta,\widetilde h]
\leq
\Prob[C=c\mid r,\theta,\widetilde h]
=r_c(r,\widetilde h).
$$
Therefore all small labels together contribute at most $2^{\ell+4n}\beta$ to the expression in equation \ref{eq1}.

\noindent  Now consider a label $c$ such that $r_c(r,\widetilde h)\geq\tau.$ Because the first stage is completed before $\Theta$ is sampled, the conditioned state is independent of the fresh function $\Theta$.  

\noindent Fix  $(r,\theta,\widetilde h)$. Let $\mathsf J$ be the following event: 
$$ J(r, \theta, \widetilde{h}):=
\left\{\max_{c',w}
\Prob[C=c',W=w\mid r,\theta,\widetilde h]
\leq\beta.\right\}
$$
If $\mathsf J$ holds, then for every accepted $w$,
$$
\begin{aligned}
&\Prob[W=w\mid C=c,R=r,\Theta=\theta,H=\widetilde h]
\\
&\qquad=
\frac{
\Prob[C=c,W=w\mid R=r,\Theta=\theta,H=\widetilde h]
}{
\Prob[C=c\mid R=r,\Theta=\theta,H=\widetilde h]
}
\\
&\qquad=
\frac{
\Prob[C=c,W=w\mid r,\theta,\widetilde h]
}{r_c(r,\widetilde h)}
\\
&\qquad\leq
\frac{\beta}{\tau}
=2^{-4n}.
\end{aligned}
$$

\noindent Fix an output value $x\in\bits^n$.  For a non-small label $c$, define
$$
p_{c,x}(r,\theta,\widetilde h)
:=
\Prob[Y=x\mid C=c,R=r,\Theta=\theta,H=\widetilde h].
$$
Define $\mathsf F_c$ to be the following event:
$$
\mathsf F_c(r,\theta,\widetilde h)
:=
\left\{
\max_w
\Prob[W=w\mid C=c,r,\theta,\widetilde h]
>2^{-4n}
\right\}.
$$
 Then it implies that
 $\mathsf F_c\subseteq \neg J$ for every such $c$.
 Set $\varepsilon_{\text{flat}}:= \Prob_{\Theta}[\mathsf F_c].$  For fixed $(r, \tilde{h},c)$ we can apply Lemma~\ref{lem:fixed-output} to the second-stage algorithm initialized in the residual state conditioned on $(R,H,C)= (r, \tilde{h},c)$. The initial state is independent of the uniform function $\Theta$, while $\tilde{h}$ and $r$ are fixed auxiliary data. Thus Lemma~\ref{lem:fixed-output} applies regardless of the distribution from which $H$ was sampled. Therefore, 
$$
\Prob_{\Theta}\!\left[
p_{c,x}(r,\Theta,\widetilde h)
\geq K_q2^{-n}
\right]
\leq
3q\exp(-2^{n+1})
+
\Prob_{\Theta}[\mathsf F_c(r,\Theta,\widetilde h)].
$$
Then we multiply $\Prob_{\Theta}[\mathsf F_c(r,\Theta,\widetilde h)]$ by $r_c$, sum over non-small labels $c$, and average over $(R,H)$ to obtain:
$$
\begin{aligned}
\Exp_{R,H,\Theta}\!\left[
\sum_{c:r_c\geq\tau}
r_c
\mathbf{1}[\mathsf F_c]
\right].
\end{aligned}
$$
If $r_c\geq\tau$ and $\mathsf F_c$ occurs, then $\mathsf J$  must fail by definition.  Hence, 
$$
\sum_{c:r_c\geq\tau}
r_c
\mathbf{1}[\mathsf F_c]
\leq
\mathbf{1}[\neg\mathsf J]
\sum_c r_c
=
\mathbf{1}[\neg\mathsf J].
$$
Therefore, for every fixed $x$,
$$
\begin{aligned}
&\Exp_{R,\Theta,H}\!\left[
\sum_{c:r_c(R,H)\geq\tau}
r_c(R,H)
\mathbf{1}[
p_{c,x}(R,\Theta,H)\geq K_q2^{-n}
]
\right]
\\
& \qquad\leq 3q\exp(-2^{n+1})+ \Pr_{R,\Theta, H}[\neg J]\\
&\qquad\leq
3q\exp(-2^{n+1})+\varepsilon.
\end{aligned}
$$
For fixed $(r,\theta,\widetilde h)$ and a non-small label $c$, define $\mathsf B_c$ to be the following event
$$
\mathsf B_c
:=\left\{\max_{x\in\bits^n}p_{c,x}(r,\theta,\widetilde h)
\geq K_q2^{-n}\right\}
.$$
Then
$$
\max_x p_{c,x}(r,\theta,\widetilde h)
\leq
K_q2^{-n}+
\mathbf{1}[\mathsf B_c].
$$
Furthermore,
$$
\mathbf{1}[\mathsf B_c]
\leq
\sum_{x\in\bits^n}
\mathbf{1}[
p_{c,x}(r,\theta,\widetilde h)
\geq K_q2^{-n}
].
$$
Since
$$
\Prob[C=c,Y=x\mid r,\theta,\widetilde h]
=r_c(r,\widetilde h)p_{c,x}(r,\theta,\widetilde h),
$$
the contribution of all non-small labels is at most
$$
\begin{aligned}
&\sum_{c:r_c\geq\tau}
\max_x\Prob[C=c,Y=x\mid r,\theta,\widetilde h]
\\
&\qquad=
\sum_{c:r_c\geq\tau}
r_c(r,\widetilde h)
\max_x p_{c,x}(r,\theta,\widetilde h)
\\
&\qquad\leq
K_q2^{-n}
\sum_{c:r_c\geq\tau}r_c(r,\widetilde h)
+
\sum_{c:r_c\geq\tau}r_c(r,\widetilde h)
\mathbf{1}[\mathsf B_c]
\\
&\qquad\leq
K_q2^{-n}
+
\sum_{x\in\bits^n}
\sum_{c:r_c\geq\tau}
r_c(r,\widetilde h)
\mathbf{1}[
p_{c,x}(r,\theta,\widetilde h)\geq K_q2^{-n}
].
\end{aligned}
$$
Then

\begin{align*}
&\Exp_{R,\Theta,H}\!\left[
\sum_{c:r_c\geq\tau}
\max_x\Prob[C=c,Y=x\mid R,\Theta,H]
\right]
\\ &\qquad\leq K_q2^{-n}+ \sum_{x\in \{0,1\}^n}\Exp_{R,\Theta,H}\!\left[
\sum_{c:r_c(R,H)\geq\tau}
r_c(R,H)
\mathbf{1}[
p_{c,x}(R,\Theta,H)\geq K_q2^{-n}
]
\right]\\
&\qquad\leq
K_q2^{-n}
+2^n\left(3q\exp(-2^{n+1})+\varepsilon\right).
\end{align*}
Finally, adding the small label contribution from equation \ref{eq:small_label} completes the proof. 

\end{proof}
\noindent For Lemma 17 and Corollary 1, $H$ is sampled as a uniformly random oracle, independently of $R$; the function $\Theta$ remains fresh, uniform, and independent.
\begin{lemma}\label{lem:lem4_expsmallthreshold}
    Let $(\mathsf{MinEnt.Prove},\mathsf{MinEnt.Verify})
$
be the  proof of min-entropy supplied by   \cref{thm:YZ-proof-of-min-entropy}. Let 
\[
a(\widetilde{h}):= \Pr[\mathsf{MinEnt.Verify}^{\widetilde{h}}(1^{\lambda^\star}, 1^h, \pi)\neq \bot: \pi\leftarrow \mathsf{MinEnt.Prove}^{\widetilde{h}}(1^{\lambda^\star},1^h)]
\]
and define 
\[
\kappa(\lambda):= \mathbb{E}_{H}[1-a(H)], \text{Good}:= \{\widetilde{h}: a(\widetilde{h})\geq \frac{1}{2}\}
\]
Then 
\begin{equation}\label{eq:kap}
\Pr_{H}[H \notin \text{Good}]\leq 2\kappa(\lambda)\leq   \negl(\lambda)
\end{equation}
\end{lemma}
\begin{proof}
    Correctness of $(\mathsf{MinEnt.Prove},\mathsf{MinEnt.Verify})
$ makes $\kappa$ negligible. Markov's inequality implies that
    \begin{align*}
       & \Pr_{H}[H \notin \text{Good}]\leq 2\kappa(\lambda)\\
        &\leq \negl(\lambda)
    \end{align*}
\end{proof}
\begin{corollary}[Fresh-slice bound for the labelled construction]\label{cor:fresh-construction}
Suppose the two-stage algorithm above is followed by the  verifier from Theorem \ref{thm:YZ-proof-of-min-entropy}, instantiated at the parameters of Construction~\ref{constr:maxent}.  Then
$$
\Exp_{R,\Theta,H}\!\left[\mathbbm 1_{H \in \text{Good}} \cdot \sum_c\max_x
\Prob[C=c,Y=x\mid R,\Theta,H]
\right]
\leq
K_q2^{-n}+\eta,
$$
where
$$ 
\eta := 2^{\ell+4n+1-h}+2^n \nu(\lambda^\star)^h+ 3q\,2^n\exp(-2^{n+1})
$$
 
For all sufficiently large $\lambda$, 
$$
\eta
\leq
(2^n+1)2^{-\lambda^{\star}}
+3q\,2^n\exp(-2^{n+1}).
$$
\end{corollary}

\begin{proof}
Fix any pair $(r,\theta)$.  Hardwire $R=r$ and $\Theta=\theta$ into the two-stage algorithm. The resulting procedure is a polynomial-query adversary, $B_{r, \theta}$ against $(\mathsf{MinEnt.Prove},\mathsf{MinEnt.Verify})
$: it outputs $(C,\Pi)$, after which the deterministic  verifier computes $W$. Construct another adversary $\hat{B}_{r, \theta}$ as follows:
\begin{enumerate}
    \item With probability $1/2$, run $B_{r, \theta}$ and return $(C,\Pi)$. 
    \item With probability $1/2$
, run $\pi_{\text{hon}}\leftarrow \mathsf{MinEnt.Prove}$, and output $(0^{\ell},\pi_{\text{hon}}).$
\end{enumerate}
Let $(\hat{C}, \hat{\Pi})$ be the output of  $\hat{B}_{r, \theta}$, let $\hat{W}:= \mathsf{MinEnt.Verify}^{H}(1^{\lambda^\star}, 1^h, \hat{\Pi})$, and let $\hat{p}(\widetilde{h})$ be the acceptance probability for a fixed $\widetilde{h}$. On $\widetilde{h}\in \text{Good}, \hat{p}(\widetilde{h})\geq \frac{a(\widetilde{h})}{2}\geq \frac{1}{4}$.  
Next, for every $c$ and accepting $w$, let $b_{c,w}(r, \theta, \widetilde{h}):= \Pr[C=c, W=w|r, \theta, \widetilde{h}]$. Then, 
\begin{align*}
    &\Pr[\hat{C}=c, \hat{W}=w|\hat{W}\neq \bot,H= \widetilde{h}]\\
    &= \frac{\Pr[\hat{C}=c, \hat{W}=w|H= \widetilde{h}]}{\hat{p}(\widetilde{h})}\\
    &\geq \frac{b_{c,w}(r, \theta, \widetilde{h})}{2\hat{p}(\widetilde{h})}\\
    &\geq \frac{b_{c,w}(r, \theta, \widetilde{h})}{2}
\end{align*}
where the last inequality follows because $\hat{p}(\widetilde{h})\leq 1$. Now, set $\beta= 2^{
1-h}$. If $H \in \text{Good}$, then $\hat{p}(\widetilde{h})\geq \frac{1}{4}$, so $\hat{B}_{r, \theta}$
 has acceptance probability at least $1/4$. If, additionally, $\max_{c,w}b_{c,w}>\beta$, then there exists a pair $(c,w)$ such that 
 $$\Pr[\hat{C}=c, \hat{W}=w|\hat{W}\neq \bot,H= \widetilde{h}]\geq 2^{-h}.$$ Applying \Cref{lem:Pr-X-and-C-is-small} with $\delta=\frac{1}{4}$ implies that 
 \[
 \Pr_{H}[H \in \text{Good } \wedge \max_{c,w} b_{c,w}(r, \theta, \widetilde{h})> \beta]\leq \nu(\lambda^\star)^h
 \]
 Averaging over $R, \theta$ preserves this inequality. Then
 \[
  \Pr_{(R, \Theta, H)}\left[  \max_{c,w} b_{c,w}(R, \Theta, H)> \beta\middle|H \in \text{Good }\right]\leq \frac{\nu(\lambda^\star)^h}{\Pr[H \in \text{Good}]}= \epsilon_{\text{good}}
 \]

Now, we can apply Lemma \ref{lem:fresh-slice}.  
\begin{align*}
&\Exp_{R,\Theta,H}\!\left[
\sum_c\max_x
\Prob[C=c,Y=x\mid R,\Theta,H]\middle|H \in \text{Good}
\right]   \\
& \quad \leq K_q2^{-n}
+2^{\ell+4n}\cdot \beta+2^n\left(3q\exp(-2^{n+1})+\epsilon_{\text{good}}\right)
\end{align*}
This implies that 
\begin{align*}
&\Exp_{R,\Theta,H}\!\left[\mathbf{1}_{H \in \text{Good}}
\sum_c\max_x
\Prob[C=c,Y=x\mid R,\Theta,H]
\right]\\
&= \Pr[H \in \text{Good}]\cdot \Exp_{R,\Theta,H}\!\left[
\sum_c\max_x
\Prob[C=c,Y=x\mid R,\Theta,H]\middle|H \in \text{Good}
\right]\\ 
&\leq  \Pr[H \in \text{Good}]\Big(K_q2^{-n}
+2^{\ell+4n+1-h}+2^n\left(3q\exp(-2^{n+1})+\epsilon_{\text{good}}\right)\Big)\\
&\leq K_q2^{-n}
+2^{\ell+4n+1-h}+2^n\left(3q\exp(-2^{n+1})+\nu(\lambda^{\star})^{h}\right)\\
&\leq  K_q2^{-n}+ \eta
\end{align*}
Finally, the parameter choice gives $2^{\ell+4n+1-h}= 2^{-\lambda^{\star}-1}$ and $\nu(\lambda^{\star})^{h}\leq 2^{-h}\leq 2^{-\lambda^{\star}}$.
This proves the corollary.
\end{proof}
\subsubsection{Block Measure-and-Reprogram}
So far, we have assumed that $C$ is fixed before the fresh function $\Theta$ is sampled.  In the real construction, the adversary has quantum access to the entire labelled oracle $G$ and may choose $C$ only after those queries.  We now prove a block measure-and-reprogram lemma that converts the real execution into the required two-stage form.

Let
$$
G:\cC\times\cW\longrightarrow\cY
$$
be an oracle.  For $c\in\cC$ and a function $\Theta:\cW\rightarrow\cY$, define 
\[
(G*\Theta_c)(c',w)
:=
\begin{cases}
\Theta(w),&c'=c,\\
G(c',w),&c'\neq c.
\end{cases}
\]
Thus $G*\Theta_c$ agrees with $G$ outside the entire block $\{c\}\times\cW$ and replaces that whole block by the fresh function $\Theta$.

\subsubsection{Oracle-query model and block projectors}

\begin{lemma}[Block measure-and-reprogram]\label{lem:block-mr}
Let $A^{G,H}$ be a $Q$-query quantum algorithm that outputs a classical label $\widehat C$ and a classical output $\widehat Z$. 
There exists a two-stage simulator $S^A$ with the following properties.
\begin{enumerate}
    \item Stage 1: $S^A$ has access to independent random oracles $(G_0,H)$ and either outputs a classical label $C$ or aborts.  If it does not abort, then only after $C$ has been fixed, it receives quantum oracle access to an independent uniform function $\Theta:\cW\longrightarrow\cY$.
    \item Stage 2: It continues the execution, answering every subsequent query by querying the oracle $G^{\star}:=G_0*\Theta_C$.
\end{enumerate}

Let $V(g,h,c,z)$ be any classical event. Then
$$
\begin{aligned}
&\Prob\!\left[
V(G^{\star},H,C,Z):
(C,Z) \leftarrow\langle S^{A,G_0,H},\Theta\rangle
\right]
\\
&\qquad\geq
\frac{1}{(2Q+1)^2}
\Prob\!\left[
V(G,H,C,Z):
(C,Z) \leftarrow A^{G,H}
\right].
\end{aligned}
$$
The probabilities average over independent uniform $G,H,\Theta$, the simulator's random choice, and all algorithmic randomness and measurements.
\end{lemma}
\begin{proof}
    The proof follows essentially verbatim from \cite{DFM20}. Let $q_g \leq Q$ be the number of queries $A$ makes to $G$. A standard query to $G$ is modelled by 
    \[
    O_{G}\ket{c}_{C_{\text{qry}}}\ket{w}_{W_{\text{qry}}}\ket{y}_{\text{out}}\ket{t}_{\text{work}}\rightarrow \ket{c,w,y\oplus g(c,w),t}
    \]
    On the $G$ query registers, define the block projector 
    \[
    P_c^{\text{qry}}= \ketbra{c}{c}_{C_{\text{qry}}}\otimes I_{W_{\text{qry}},\text{out},\text{work}}
    \]
    On the final output-label register $\widehat C$, define
    \[
    P_c^{\text{out}}= \ketbra{c}{c}_{\widehat C} \otimes I
    \]
  Now fix $g_0 , \theta, c$, and let $g^{\star}= g_0 * \theta_c$. 
  Crucially, because $g^{\star}$ and $g_0$ agree outside the block $c\times \cW$, their oracle unitaries satisfy
  \begin{equation}\label{eq:measure_reprog}
  O_{g_0}(I-P_c^{\text{qry}})= O_{g^{\star}}(I-P_c^{\text{qry}})
  \end{equation}
  Now we can apply the proof of Lemma 1 and Theorem 2 in \cite{DFM20} almost verbatim, replacing their rank-one query projector $\ketbra{x}{x}$ with $P_c^{\text{qry}}$, their final-output projector by $P_c^{\text{out}}$, and fresh oracle value by the fresh block function $\theta_c: \cW\longrightarrow\cY$. Their proof uses only that the relevant operator is an orthogonal projector and that the original and reprogrammed oracle unitaries agree on its orthogonal complement, which is exactly Equation \ref{eq:measure_reprog}. So we can conclude that 
  $$
\begin{aligned}
&\Prob\!\left[
V(G^{\star},H,C,Z):
(C,Z) \leftarrow\langle S^{A,G_0,H},\Theta\rangle
\right]
\\
&\qquad\geq
\frac{1}{(2q_g+1)^2}
\Prob\!\left[
V(G,H,C,Z):
(C,Z) \leftarrow A^{G,H}
\right].\\
&\qquad\geq
\frac{1}{(2Q+1)^2}
\Prob\!\left[
V(G,H,C,Z):
(C,Z) \leftarrow A^{G,H}
\right].
\end{aligned}
$$
\end{proof}
\subsubsection{The adaptive-selector labelled hashing lemma}
Let 
\[
( C,  \Pi)\leftarrow A^{G,H}(1^{\lambda}, 1^n)
\]
and define 
\[
W\leftarrow \mathsf{MinEnt}.\Verify^H(1^{\lambda^\star}, 1^h, \Pi)
\]
and 
\[
X:= \begin{cases}
    G( C, W) &\text{ if }  W \neq \bot\\
    \bot &\text{ if }  W = \bot
\end{cases}
\]
For a fixed oracle pair $(g,\widetilde h)$, define $
p(g,\widetilde h)
:=
\Pr[W\neq\bot\mid G=g,H=\widetilde h].$

\noindent If $p(g,\widetilde h)>0$, let

$$
\Gamma(g,\widetilde h)
:=
\sum_{c\in\{0,1\}^{\ell}}
\max_{x\in\{0,1\}^{n}}
\Pr[C=c,X=x\mid W\neq\bot,G=g,H=\widetilde h],
$$

\noindent and otherwise set $\Gamma(g,\widetilde h):=0$. Thus, by definition, whenever $p(g,\widetilde h)>0$,

$$
\Gamma(g,\widetilde h)
=
2^{-H^{\mathrm{avg}}_{\infty}(X\mid C,W\neq\bot,G=g,H=\widetilde h)}.
$$

\noindent For given values of $(g, \widetilde{h}, c)$ for which $p(g, \widetilde{h}) > 0$, let $\chi_{g,\widetilde h}(c)$ be the lexicographically first value of $x$ that maximizes $\Pr[C=c,X=x\mid W\neq\bot,G=g,H=\widetilde h]$.

\noindent Thus, 
for every fixed $(g,\widetilde h)$ such that $p(g, \widetilde{h}) > 0$,

\begin{align*}
&\Pr[W\neq\bot,\ X=\chi_{g,\widetilde h}(C)\mid G=g,H=\widetilde h]\\
&= \sum_{c}\Pr[C=c,W\neq\bot,\ X=\chi_{g,\widetilde h}(c)\mid G=g,H=\widetilde h]\\
&= p(g, \widetilde h)\sum_{c}\Pr[C=c,\ X=\chi_{g,\widetilde h}(c)\mid W\neq\bot,G=g,H=\widetilde h]\\
&= p(g, \widetilde h)\sum_{c}\max_x\Pr[C=c,\ X=x\mid W\neq\bot,G=g,H=\widetilde h]\\
&= p(g, \widetilde h)\cdot \Gamma(g, \widetilde h)
\end{align*}

\noindent Let
\[
q:=Q+1,
\quad
K_q:=2^6(q+1)^3=2^6(Q+2)^3,
\quad
\eta:=(2^n+1)2^{-\lambda^{\star}}+3q\,2^n\exp(-2^{n+1}).
\]

\begin{lemma}\label{lem:final}
For every $\lambda$, every  $u(\lambda)\in(0,1]$, and every inverse-polynomial $\delta= \delta(\lambda)$,
\begin{equation}
\label{eq:1}
\Pr_{G,H}[p(G,H)\geq\delta\ \wedge\ \Gamma(G,H)\geq u]
\leq 2\kappa(\lambda)+
\frac{(2Q+1)^2}{\delta u}\bigl(K_q2^{-n}+\eta\bigr).
\end{equation}
\end{lemma}

\begin{proof}
Let

$$
\mathsf{Bad}:=\{(g,\widetilde h):\widetilde{h}\in \text{Good}, p(g,\widetilde h)\geq\delta,
\ \Gamma(g,\widetilde h)\geq u\},
\qquad
\rho:=\Pr[(G,H)\in\mathsf{Bad}].
$$

Then
\begin{align*}
&\Pr[(G,H)\in\mathsf{Bad},\ W\neq\bot,\ X=\chi_{G,H}(C)]\\
&= \Pr[\ W\neq\bot,\ X=\chi_{G,H}(C)|(G,H)\in\mathsf{Bad}]\cdot \Pr[(G,H)\in\mathsf{Bad}]\\
&
\geq
\delta \cdot u\cdot \rho.
\end{align*}
Now we will apply 
Lemma \ref{lem:block-mr} to the oracle $G$, taking the auxiliary output $\widehat Z$ to be $\Pi$. The predicate $V(g,\widetilde h,c, \pi)$ is as follows:
\begin{enumerate}
    \item $(g,\widetilde h)\in \mathsf{Bad}$. 
    \item $W \neq \bot$.
    \item $X = \chi_{g,\widetilde{h}}(c)$.
\end{enumerate}
Let $\mathsf{Succ}$ denote the simulator's success event. Then

$$
\Pr[\mathsf{Succ}]\geq\frac{\rho\delta u}{(2Q+1)^2}.
$$
Note that if $\mathsf{Succ}$ occurs, then $(g, \widetilde{h}) \in \mathsf{Bad}$ and $\widetilde{h} \in \mathsf{Good}$.

We now upper-bound $\Pr[\mathsf{Succ}]$.
Let $G_0$ denote the complete base outer oracle sampled before the block measurement.  The simulator's first stage uses $(G_0,H)$ and outputs $C$ before receiving an independent fresh function
$$
\Theta:\bits^m\longrightarrow\bits^n.
$$
The oracle used in the second stage is 
$$
G^{\star}=G_0*\Theta_C,
$$
In particular, 
$$
G^{\star}(C,\cdot)=\Theta.
$$
 
\noindent Let $R:=G_0$ in the experiment described in section \ref{sec:lab}. Conditional on $(G_0,H)$, the label is selected before $\Theta$ is sampled, so the required independence holds. 
Corollary~\ref{cor:fresh-construction} therefore yields
\begin{equation}\label{eq:cor1_imp}
\Exp_{G_0,\Theta,H}\!\left[\mathbf{1}_{H \in \text{Good}}
\sum_c\max_x
\Prob[C=c,\Theta(W)=x\mid G_0,\Theta,H]
\right]
\leq
K_q2^{-n}+\eta.
\end{equation}
Fix a particular triple
$$
(G_0,\Theta,H)=(g_0,\theta,\widetilde h).
$$
Observe that
$x_c
:=
\chi_{g_0*\theta_c,\widetilde h}(c)$
is fixed once $(g_0,\theta,\widetilde h,c)$ is fixed.

  Thus
\begin{align*}
\Prob[\Succ\mid g_0,\theta,\widetilde h] &= \mathbbm{1}_{\widetilde{h}\in \text{Good}} \cdot \Prob[\Succ\mid g_0,\theta,\widetilde h]\\
&= \mathbbm{1}_{\widetilde{h}\in \text{Good}} \ifSubmission{\\&\quad}\fi\cdot \sum_c
\Prob[C=c,(g_0*\theta_c ,\widetilde h) \in \mathsf{Bad}, W\neq\bot,
\Theta(W)=x_c
\mid g_0,\theta,\widetilde h]\\
&\leq
\mathbbm{1}_{\widetilde{h}\in \text{Good}} \cdot \sum_c
\Prob[C=c,W\neq\bot,
\Theta(W)=x_c
\mid g_0,\theta,\widetilde h]
\\
&\leq
\mathbbm{1}_{\widetilde{h}\in \text{Good}} \cdot \sum_c
\max_{x\in\bits^n}
\Prob[C=c,\Theta(W)=x
\mid g_0,\theta,\widetilde h].
\end{align*}

 Averaging  over $(G_0,\Theta,H)$ and applying equation \ref{eq:cor1_imp} yields:
$$
\Prob[\Succ]
\leq
K_q2^{-n}+\eta.
$$

 Combining the two bounds and solving for $\rho$ shows that $\rho \leq \frac{(2Q+1)^2}{\delta u}\bigl(K_q2^{-n}+\eta\bigr)$. Finally, every violation in equation \ref{eq:1} either belongs to $\text{Bad}$ or has $H \notin \text{Good}$. By equation \ref{eq:kap}, this is at most $\rho+ 2 \kappa(\lambda)$, proving the lemma. 
\end{proof}

\subsubsection{Putting everything together}

\begin{theorem}[Maximal conditional min-entropy]\label[theorem]{thm:adaptive-maximal-min-entropy}
 Construction \ref{constr:maxent} satisfies the soundness property in Definition \ref{def:adaptive}.
\end{theorem}

\begin{proof}
 Fix polynomially bounded $n= n(\lambda)=\omega(\log\lambda),Q= Q(\lambda),\ell= \ell(\lambda)$, an inverse-polynomial $\delta$, and $\gamma=\omega(\log\lambda)$. If $\gamma>n$, the entropy event is empty. Otherwise set

$$
u:=2^{-n+\gamma}\in(0,1].
$$

By the definition of $\Gamma$ and Lemma \ref{lem:final},
\begin{align*}
&\Pr_{G,H}\bigl[p(G,H)\geq\delta\ \wedge\
H^{\mathrm{avg}}_{\infty}(X\mid C,W\neq\bot,G,H)\leq n-\gamma\bigr]
\\
&\quad\quad\leq 2\kappa(\lambda)+
(2Q+1)^2\delta^{-1}
\bigl(K_q2^{-\gamma}+2^{n-\gamma}\eta\bigr)\\
&\quad\quad= 2\kappa(\lambda)+(2Q+1)^2\delta^{-1}
K_q2^{-\gamma}+(2Q+1)^2\delta^{-1}2^{n-\gamma}\eta\\
&\quad\quad=  2\kappa(\lambda)+(2Q+1)^2\delta^{-1}
K_q2^{-\gamma}\\
&\quad\quad\quad\quad+(2Q+1)^2\delta^{-1}2^{n-\gamma}\cdot \Big((2^n+1)2^{-\lambda^{\star}}+3q\,2^n\exp(-2^{n+1})\Big)
\end{align*}

The term $2\kappa(\lambda)$ is negligible because of correctness of the proof of min-entropy from Theorem \ref{thm:YZ-proof-of-min-entropy}.  The term with $K_q$ is negligible because $(2Q+1)^2\delta^{-1}K_q$ is polynomially bounded and $\gamma=\omega(\log\lambda)$. Moreover, using $\lambda^{\star}=\lambda+3n+\ell+10$ and $2^n+1\leq2^{n+1}$,

$$
2^{n-\gamma}(2^n+1)2^{-\lambda^{\star}}
\leq
2^{-\lambda-n-\ell-\gamma-9}.
$$

For all sufficiently large $n$, $2^{2n}\exp(-2^{n+1})\leq\exp(-2^n)$, and therefore

$$
3q\,2^{2n-\gamma}\exp(-2^{n+1})
\leq
3q\exp(-2^n).
$$

Both bounds are negligible; the latter uses $n=\omega(\log\lambda)$. Multiplication by $(2Q+1)^2\delta^{-1}$ preserves negligibility. Finally,

$$
p(G,H)=\Pr[W\neq\bot\mid G,H]=\Pr[X\neq\bot\mid G,H],
$$

so this is exactly adaptive-label soundness.
\end{proof}

\begin{corollary}[Maximal entropy rate]
For every $\gamma(\lambda)=\omega(\log\lambda)$, the accepted output has average conditional min-entropy greater than $n-\gamma(\lambda)$ given the adaptive label, except with negligible probability over the random oracles whenever acceptance has inverse-polynomial probability.
\end{corollary}
\fi

\section{Acknowledgements}
We thank Dakshita Khurana and James Bartusek for valuable discussions.

\bibliographystyle{alpha}
\bibliography{references}

\ifSubmission
    \appendix
    
\fi

\end{document}